\documentclass[a4paper, reqno, 11pt]{amsart}
\usepackage{amsmath}
\usepackage{amssymb}
\usepackage{bm}
\usepackage{graphicx,color}
\usepackage{amsbsy}
\usepackage{amsmath}
\usepackage{amsfonts}
\usepackage{amsthm}
\usepackage{listings}
\usepackage{enumerate}
\usepackage[utf8]{inputenc}
\usepackage{stackrel}

\usepackage{enumitem}

\usepackage{geometry} 
\newcommand{\ce}{c}
\newcommand{\gini}{C}

\newcommand{\R}{\mathbb{R}}
\newcommand{\C}{\mathbb{C}}
\newcommand{\N}{\mathbb{N}}

\newcommand{\Tr}{{\rm Tr}}

\newcommand{\lrb}{\left(}
\newcommand{\rrb}{\right)}

\numberwithin{equation}{section}

\theoremstyle{plain}
\newtheorem{theorem}{Theorem}[section]
\newtheorem{lemma}[theorem]{Lemma}
\newtheorem{corollary}[theorem]{Corollary}
\newtheorem{proposition}[theorem]{Proposition}

\theoremstyle{definition}
\newtheorem{remark}[theorem]{Remark}
\newtheorem{definition}[theorem]{Definition}
\newtheorem{example}[theorem]{Example}

\newcommand{\eqind}[1]{\begin{align}#1\end{align}}
\newcommand{\eqn}[1]{\begin{align*}#1\end{align*}}
\newcommand{\eq}[1]{\begin{align}\begin{aligned}#1\end{aligned}\end{align}}
\newcommand{\eqtext}[1]{\qquad \text{#1} \qquad}

\DeclareMathOperator{\trace}{\mathrm{Tr}} 

\begin{document}

\title{
Additivity violation of quantum channels \\via
strong convergence \\ to semi-circular and circular elements
}

\author{Motohisa Fukuda}
\address{MF: Yamagata University, 1-4-12 Kojirakawa, Yamagata, 990-8560 Japan}
\email{fukuda@sci.kj.yamagata-u.ac.jp}

\author{Takahiro Hasebe}
\address{TH: Department of Mathematics, Hokkaido University,
Kita 10, Nishi 8, Kita-Ku, Sapporo, Hokkaido, 060-0810, Japan}
\email{thasebe@math.sci.hokudai.ac.jp}

\author{Shinya Sato}
\address{SS:}
\email{hysmk6331@gmail.com}

\subjclass[2010]{46L54,81P45}

\keywords{quantum channels, minimum output entropy, additivity violation, Gaussian unitary ensemble, Ginibre ensemble, free probability, Haagerup inequality}

\maketitle
\markright{\MakeUppercase{Additivity violation of quantum channels with semi-circular and circular elements}}

\begin{abstract}
Additivity violation of minimum output entropy, which shows non-classical properties in quantum communication, had been proved in most cases for random quantum channels defined by Haar-distributed unitary matrices.  
In this paper, we investigate random completely positive maps made of Gaussian Unitary Ensembles and Ginibre Ensembles regarding this matter.
Using semi-circular systems and circular systems of free probability, we not only show the multiplicativity violation of maximum output norms in the asymptotic regimes but also prove the  additivity violation via Haagerup inequality for a new class of random quantum channels constructed by rectifying the above completely positive maps based on strong convergence.
\end{abstract}

\tableofcontents

\section{Motivation from quantum information and our results}

Let $\hat D_n$ be the set of $n \times n$ positive semi-definite Hermitian matrices, and $D_n$ the set of $n \times n$ positive semi-definite Hermitian matrices of trace one. The latter are called quantum states. For a completely positive (CP) map $\Phi : D_\ell \to  \hat D_n $ we define the maximum output $p$-norm (MO$p$N): 
\eq{
\| \Phi \|_p = \max_{\rho \in D_\ell} \| \Phi(\rho) \|_p 
}
where $\|\cdot\|_p$ is the $p$-th Schatten norm.
Note that norm convexity implies that an optimal input $\rho$ can be assumed to be a rank-one projection.
Moreover, if a CP map preserves trace, it is called quantum channel. In this case, one can define minimum output entropy (MOE):
\eq{
S_{\min} (\Phi) = \min_{\rho \in \mathcal D_\ell} S (\Phi(\rho))
}
where 
$
S(\sigma) = - \trace [\sigma \log \sigma]
$ 
is the von Neumann entropy of a quantum state $\sigma$.
Again, concavity of entropy implies that an optimal input $\rho$ can be assumed to be a rank-one projection.

For those quantities, one can ask the following questions of multiplicativity and additivity \cite{AmoHolWer00, king2001minimal}:
\eq{\label{eq:question}
\|\Phi \otimes \Omega \|_p \stackrel{?}{=} \|\Phi \|_p \cdot \|\Omega \|_p
\quad \text{and} \quad 
S_{\min} (\Phi \otimes \Omega) \stackrel{?}{=} S_{\min} (\Phi) + S_{\min} (\Omega) \ .
}
Since the following inequalities are obviously true:
\eq{\label{eq:obvious_inequality}
\|\Phi \otimes \Omega \|_p  \geq  \|\Phi \|_p \cdot \|\Omega \|_p
\quad \text{and} \quad 
S_{\min} (\Phi \otimes \Omega)  \leq  S_{\min} (\Phi) + S_{\min} (\Omega) \ ,
}
the above questions in \eqref{eq:question} are actually whether the inequalities in \eqref{eq:obvious_inequality} are strict or not. 
Although some positive examples were found, see for example 
\cite{king2002additivity, king_maxp, king2003capacity, shor2002additivity},
the equalities in \eqref{eq:question} do not hold in general. 
The existence of quantum channels which violate multiplicativity or additivity was proved by using randomly generated quantum channels, first in \cite{hw_vio} and \cite{has}, respectively.
Later, some more detailed or general proofs followed  \cite{fukuda2010comments, brandao2010hastings, fukuda2010entanglement}. Further, different proofs were also made via asymptotic geometric analysis  \cite{aubrun2010nonadditivity, aubrun2011hastings, f-vio} and via free probability \cite{belinschi2016almost, collins_random_unitary, collins2015estimates, collins2020concentration}. Recently this problem showed connection to black hole entropy \cite{hayden2020black}. 

In the breakthrough results \cite{hw_vio, has}, they showed that such violations are asymptotically typical by using the following randomly defined quantum channels,
which are called random unitary channels or random external fields:
\eq{\label{eq:ruc}
\mathcal E (\rho) = \sum_{i=1}^k p_i U_i \rho U_i^* \ .
}
Here, $\{U_i\}_{i=1}^k \subset \mathcal U(n)$ are i.~i.~d.~Haar-distributed unitary matrices with $1 \ll k \ll n$.

In this paper we replace those unitary matrices by i.~i.~d.~Gaussian Unitary Ensembles (GUEs) $\{S_i\}_{i=1}^k$ or i.~i.~d.~Ginibre Ensembles (GEs) $\{\gini_i\}_{i=1}^k$, i.e.
\eq{\label{eq:our-map}
\Phi_n(\rho) =\frac 1k \sum_{i=1}^k X_i \rho X_i^*, 
}
where $X_i = S_i$ for all $i \in [k]:=\{1,2,\dots,k\}$ or $\gini_i$ for all $i \in [k]$. 
Here, each $n \times n$ Hermitian matrix $S$ for GUEs is defined by the following i.~i.~d.~entries: for $i<j$, 
\eq{
(S)_{i,i}, \quad \sqrt{2} \mathrm{Re}(S)_{i,j}, \quad \sqrt{2}\mathrm{Im}(S)_{i,j} \sim N(0,1/n) \ ,
}
and GEs are then defined by the relation 
\eq{\label{eq:gini_circ}
\gini_i = \frac{S_i + \sqrt{-1} \, S_i^\prime}{\sqrt{2}}\ , 
}
where $\{S_i,S_i'\}_{i\in[k]}$ are i.~i.~d.~GUEs. Note that each GE above has entries whose real parts and imaginary parts are all independent and distributed as $N(0,1/(2n))$. 
Then, by using free probability we explore not only the violation of multiplicativity of MO$p$N but also violation of additivity of MOE for quantum channels made out of these CP maps, as stated below.

Let us summarize our results briefly here. We show first typicality of violation of multiplicativity: for $1.5<p\leq \infty$ and $1 \ll k \ll n$
\eq{\label{eq:vio_multi}
\|\Phi_n \otimes \bar\Phi_n \|_p > \|\Phi_n \|_p \cdot \|\bar\Phi_n\|_p
}
where we defined the complex conjugate of $\Phi_n$ by
\eq{\label{eq:our-map_conjugate}
\bar\Phi_n(\rho) =\frac 1k \sum_{i=1}^k \bar X_i \rho X_i^T \ .
}
Note that $\overline{X^*} = X^T$.
One can find more precise statements in Theorem \ref{theorem:violation-cp}, where 
we also identify the shape of optimal output eigenvalue distributions of $\Phi_n$ in the asymptotic regime. 
In \cite{belinschi2016almost}, exact calculations were made to give precise estimates on such problems for different random quantum channels to show larger additivity violation of MOE than before.

These maps $\Phi_n$ in \eqref{eq:our-map} are not trace-preserving but nearly so almost surely when $k$ and $n$ are large.
Hence, we can rectify $\Phi_n$ to
create quantum channels $\Psi_n$, and prove typicality of violation of additivity: for $1 \ll k \ll n$,
\eq{\label{eq:vio_add}
S_{\min} (\Psi_n\otimes \bar \Psi_n) < S_{\min} (\Psi_n) + S_{\min} (\bar \Psi_n) \ .
}
This is our second result, which is stated in detail in Theorem \ref{theorem:violation-qc}.
To prove the claim we adopted the strategy in
 \cite{collins_random_unitary}, where the Haagerup inequality was first applied to additivity questions.
In addition, the above violation of multiplicativity is extended to these quantum channels in Theorem \ref{theorem:violation-cp2}.
In case of GEs, the quantum channels $\Psi_n$ have the same distribution as the well-known random quantum channels; see the end of Section \ref{sec:rectify}.

Those questions in \eqref{eq:question} were made in the field of quantum information theory. Classical communication theory was established by C. Shannon \cite{shannon1948mathematical}, but it does not extend to quantum communication because of the quantum correlation called entanglement. In fact the questions in \eqref{eq:question} are about whether or not entanglement improves MO$p$N or MOE, and are related to an operational quantity called capacity \cite{Holevo1998, SchumacherWestmoreland1997}, through the proven facts in \cite{shor2004equivalence}.
Also, note that the von Neumann entropy is the limit of R\'enyi $p$-entropy as $p \to 1$, and multiplicativity of MO$p$N is equivalent to additivity of R\'enyi $p$-entropy.

This paper is organized as follows. First, in Section \ref{sec:free_prob} necessary facts in free probability are provided.
Next, then violation of multiplicativity is proved in Section \ref{sec:vio-cp} by investigating the limiting output eigenvalues of the CP maps in \eqref{eq:our-map}.
Finally, additivity violation for quantum channels which are made by rectifying the CP maps, is proved in Section \ref{sec:vio-qc}.

\section{Free probability and polynomials in GUEs and GEs}\label{sec:free_prob}

Free probability has applications to the study of eigenvalues of random matrices as the size of the matrices tends to infinity. The first main result was by Voiculescu \cite{voiculescu1991} who found that a sequence of tuples of independent GUEs is asymptotically free.  The asymptotic freeness implies that the empirical eigenvalue distribution of self-adjoint polynomials in the independent GUEs weakly converges to a probability measure almost surely, and the limit distribution can be described by free probabilistic methods, e.g.\ the $R$-transform and $S$-transform.   
Note that asymptotic freeness also holds for rotationally invariant random matrices, see e.g.\ \cite{nica2006lectures, mingo2017free}.

For applications to random quantum channels, strong convergence is important. For Hermitian random matrices, the strong convergence says that in addition to weak convergence of the empirical eigenvalue distributions to a probability measure $\mu$, the largest (resp.\ the smallest) eigenvalue converges to the maximum (resp.\ the minimum) of the support of the limiting distribution $\mu$. 
The strong convergence holds for polynomials in independent GUEs \cite{haagerup_thorbjornsen} (see also Theorem \ref{thm:strong} below) and for polynomials in independent Haar unitary matrices \cite{collins_male}. 
On the other hand, the failure of the strong convergence means that some eigenvalues, called the outliers or spikes,  remain to lie outside the support of  the limit distribution $\mu$. This often occurs when random matrices are perturbed, see \cite{BBP2005,peche2006} and a more recent article \cite{BBC2017}. 

We introduce below  some necessary concepts on free probability and also prove some supplementary results.  

\subsection{Free probability}
This subsection collects some basic notions in free probability. See \cite{nica2006lectures, mingo2017free} for further details.

Let $\mathcal{A}$ be a $C^*$-algebra and $\varphi$ be a tracial faithful state on $\mathcal A$. We call the pair $(\mathcal A, \varphi)$ a $C^*$-probability space. For a self-adjoint element $a \in \mathcal A$ the unique probability measure $\mu$ on $\R$ satisfying the condition
\begin{equation}\label{eq:distribution}
\varphi(a^n) = \int_\R x^n \,d\mu_a(x), \qquad n \in \N,   
\end{equation}
is called the distribution of $a$ and denoted by $\mu_a$. In fact, $\mu_a$ is supported on the spectrum of $a$. 

For a possibly non self-adjoint element $a \in \mathcal A$, the set of its $*$-moments 
\eq{
\{\varphi(a^{\epsilon_1} a^{\epsilon_2} \cdots a^{\epsilon_n}): \epsilon_1,\epsilon_2,\dots,\epsilon_n \in \{1,\ast\}, n\in \N\}
}
is called the $*$-distribution of $a$. The $*$-distribution of $a$ can also be written as 
\eq{
\{\varphi(P(a))\}_{P(z)\in \C^*[z]} \ , 
}
where $\C^*[z]$ is the unital $*$-polynomial ring; see Section \ref{section:support}. Two elements $a$ and $b$ are said to have the same $*$-distribution if 
\eq{
\varphi(a^{\epsilon_1} a^{\epsilon_2} \cdots a^{\epsilon_n}) = \varphi(b^{\epsilon_1} b^{\epsilon_2} \cdots b^{\epsilon_n})
}
for any $ \epsilon_1,\epsilon_2,\dots,\epsilon_n \in \{1,\ast\}, n\in \N$. When $a$ and $b$ are self-adjoint, then this is equivalent to $\mu_a=\mu_b$.

A self-adjoint element $s \in \mathcal A$ is called a standard semi-circular element if $\mu_s$ is the standard semi-circle law
\begin{equation}
\frac{1}{2\pi}\sqrt{4-x^2}\, \mathbf1_{(-2,2)}(x)\,dx. 
\end{equation}
A family $\{s_i: i\in I\}$ of elements of $\mathcal A$ is called a semi-circular system if it is free and each member $s_i$ is a standard semi-circular element. 

An element $c \in \mathcal A$ is called a standard circular element if $c=\frac{s + \sqrt{-1}s'}{\sqrt{2}} $, where $\{s,s'\}$ is a semi-circular system. A family $\{\ce_i: i\in I\}$ is called a circular system if 
\eq{\label{eq:circular}
\ce_i=\frac{s_i+\sqrt{-1}s_i'}{\sqrt{2}},\qquad i\in I\ ,
}
where $\{s_i,s_i': i \in I\}$ is a semi-circular system.

A typical construction of a semi-circular system is given on  full Fock spaces, see Section \ref{sec:haagerup}. However, when we discuss semi-circular systems, the underlying $C^*$-probability space will not be specified unless needed.

The Marchenko-Pastur distribution of rate $\lambda>0$ is 
\begin{equation}
\max\{0,1-\lambda\}\delta_0(dx) + \frac{ \sqrt{(b-x)(x-a)} }{2\pi x} \,\mathbf1_{(a,b)}(x)\,dx, 
\end{equation}
where $a=(1-\sqrt{\lambda})^2$ and $b=(1+\sqrt{\lambda})^2$. 
For a standard semi-circular element $s$, computing the probability density functions yields that its square $s^2$ has the Marchenko-Pastur distribution of rate 1. For a standard circular element $c$, the self-adjoint element $c^* c$ also has the Marchenko-Pastur distribution of  rate 1; see \cite[Example 15.15]{nica2006lectures}.

For a self-adjoint element $a \in \mathcal A$, its Cauchy transform is defined by 
\begin{equation}
G_a(z) = \varphi((z-a)^{-1}) =  \int_{\text{Spec}(a)} \frac{1}{z-x}\,d\mu_a(x), \qquad z \in \C \setminus\text{Spec}(a). 
\end{equation}
where $\text{Spec}(a)$ is the spectrum of $a$. Note that a maximal domain for the analyticity of $G_a|_{\C \setminus\R}$ can identify $\sup \text{Spec}(a)$: 
\begin{equation}\label{eq:support}
\sup \text{Spec}(a) = \sup\{ t>0: \text{$G_a|_{\C \setminus\R}$ has analytic continuation to  $\C\setminus(-\infty,t]$}\}. 
\end{equation}
Furthermore, $G_a$ has the series expansion  
\begin{equation}
G_a(z)= \sum_{n=0}^\infty\frac{\varphi(a^n)}{z^{n+1}}, \qquad |z| > \|a\|. 
\end{equation}
From this $G_a$ has a compositional inverse series $G_a^{-1}(z)=\frac{1}{z} + \varphi(a) + [\varphi(a^2)-\varphi(a)^2]z + \cdots$. 
Then the R-transform is defined by
\begin{equation}
R_a(z) = G_a^{-1}(z) -\frac{1}{z} = \sum_{n=1}^\infty \kappa_n[a]z^{n-1}, 
\end{equation}
which is a convergent series in a neighborhood of $0$ of the complex plane. The coefficients $\kappa_n[a]$ are called the free cumulants of $a$. If $a,b \in\mathcal A$ are free then the additivity of R-transform
\begin{equation}\label{eq:R}
R_{a+b} = R_a + R_b
\end{equation}
holds in a neighborhood of $0$. This additivity enables us to compute $\mu_{a+b}$ in terms of $\mu_a,\mu_b$, when $a$ and $b$ are self-adjoint and free. In the language of probability measures, $\mu_{a+b}$ is called the free convolution of $\mu_a$ and $\mu_b.$ 

\begin{example}\label{ex:R}
\begin{enumerate}[label=\rm(\arabic*)]
\item For a standard semi-circular element $s$ we have $R_s(z)=z$. 

\item  Let $a\in\mathcal A$ be a self-adjoint element having the Marchenko-Pastur distribution with rate $\lambda>0$. It has the R-transform $R_a(z) = \lambda/(1-z)$. 
\end{enumerate}
\end{example}

A combinatorial description of free cumulants is also useful and allows us to define multi-variate free cumulants. 
A set partition $\pi$ of a finite set $I$ is a family of disjoint non-empty subsets, called the blocks, of $I$ so that their union equals $I$. When $I$ is linearly ordered, then a set partition is called crossing if there exist  $i,j$ in a common block and $k,\ell$ in another distinct  common block such that $i<k<j<\ell$; otherwise, a set partition is called non-crossing. For example, the set partition $\{\{1,3,5\},\{2,4\}\}$ of the finite linearly ordered set $[5]$ is crossing, while the set partition $\{\{1,2,5\},\{3,4\}\}$ is non-crossing. Let $NC(n)$ stand for the collection of all non-crossing set partitions of $[n]$. 
Some special non-crossing set partitions we need are summarized below: for $k,n \in \N$, let  
\begin{enumerate}[label=(\arabic*)]
\item $\hat{1}_{k}:=\{[k]\}\in NC(k)$;  

\item $\hat{1}^n_k :=\{\{1,2,\dots,k\},\{k+1,k+2,\dots,2k\}, \dots, \{(n-1)k +1,\dots, nk\}\} \in NC(nk)$;  

\item $NC_2(2k) = \{\pi\in NC(2k): \text{every block of $\pi$ has cardinality two}\}$.
\end{enumerate}
For a set partition $\pi=\{B_1,B_2,\dots,B_k\}$ of $[n]$ and a family of multilinear functionals  $\{L_n\colon \mathcal A^n \to \mathbb C\}_{n\ge1}$, we define the multiplicative extension
$L_\pi\colon\mathcal A^n \to \mathbb C$ by 
\[
L_\pi[a_1,a_2,\dots,a_n]= \prod_{B\in \pi} L_{|B|}[a_B], 
\]
where $a_B$ stands for the vector $(a_{i_1},a_{i_2},\dots, a_{i_k})$ for $B=\{i_1,i_2,\dots,i_k\}$ with the order $i_1<i_2<\cdots <i_k$. 

The state $\varphi$ naturally associates the multilinear functional $\varphi_n\colon \mathcal A^n \to \mathbb C$ defined by 
$\varphi_n[a_1,a_2,\dots,a_n]=\varphi(a_1a_2\cdots a_n)$. 
This further associates $\varphi_\pi$ for all non-crossing set partitions $\pi$ of $[n]$ for all $n \in \mathbb N$. 

A family of multilinear functionals $\{\kappa_n\colon\mathcal A^n\to\C\}_{n\ge1}$ is defined recursively via
\[
\varphi(a_1a_2\cdots a_n) = \sum_{\pi \in NC(n)} \kappa_\pi[a_1,\dots, a_n], 
\]
and each $\kappa_n$ is called a (multivariate) free cumulant. Some free cumulants of low orders are $\kappa_1[a]=\varphi(a)$ and $\kappa_2[a_1,a_2]=\varphi(a_1a_2)- \varphi(a_1)\varphi(a_2)$. There is an inversion formula, called the M\"obius inversion formula, that expresses $\kappa_n$ in terms of $(\varphi_\pi)_{\pi\in NC(n), n \in \mathbb N}$. The coefficients $\kappa_n[a]$ in the R-transform $R_a(z)$ are obtained as the multi-variate free cumulants evaluated at the diagonal vector $(a,a,\dots,a)$. 

\subsection{Strong convergence for GUEs and GEs}
\label{section:support}

Combining (the almost sure version of) asymptotic freeness (see e.g.\ \cite[Theorem 4 in Chapter 4]{mingo2017free}) and \cite[Theorem A]{haagerup_thorbjornsen}, the strong convergence for self-adjoint polynomials in GUEs (and then GEs) is explained below.

Let $\C^*[z_1,z_2,\dots, z_k]$ be the unital  $\ast$-polynomial ring generated by non-commutative indeterminates $z_1,z_2,\dots, z_k$; that is, it is a unital polynomial ring generated by non-commutative indeterminates $z_1,z_2,\dots, z_k, y_1,y_2,\dots, y_k$ equipped with the anti-linear mapping $*$ defined by $1^*= 1$, $z_i^*= y_i$ and $y_i^*= z_i$ for all $i$ and then extended to $\C^*[z_1,z_2,\dots, z_k]$ so that $(PQ)^*=Q^* P^*$ for all $P,Q\in\C^*[z_1,z_2,\dots, z_k]$.  Later on the element $y_i$ will always be denoted by $z_i^*$.

A polynomial $P \in \C^*[z_1,z_2,\dots, z_k]$ is said to be self-adjoint if $P^*=P$. For example, the polynomials $1, z_1z_2^*z_2z_1^*, z_1 +z_1^*$ are self-adjoint, while $z_1, z_1^* z_2 z_1 $ are not.

\begin{theorem}[Strong convergence] \label{thm:strong}
Suppose we have $k$ independent $n \times n$ GUEs (or GEs) $X_1, \ldots, X_k$. Let $x_1,\dots, x_k$ be a semi-circular (or circular) system. For a self-adjoint polynomial $P(z_1,z_2,\cdots, z_k) \in \C^*[z_1,z_2,\dots, z_k]$ the following 
assertions hold as $n\to\infty$. 
\begin{enumerate}[label=\rm(\arabic*)]
\item The empirical eigenvalue distribution of the random matrix $P(X_1,X_2,\dots, X_k)$ converges weakly to $\mu_{P(x_1,x_2,\dots, x_k)}$ almost surely.  
\item The largest and smallest eigenvalues of $P(X_1,X_2,\dots, X_k)$ converge to $\sup \text{\normalfont Spec}(P(x_1,x_2,\dots, x_k))$ and $\inf \text{\normalfont Spec}(P(x_1,x_2,\dots, x_k))$ almost surely, respectively. 
\end{enumerate}
\end{theorem}
Note that GEs $\{C_i\}_{i\in I}$ and a circular system $\{\ce_i\}_{i\in I}$ have the respective relations \eqref{eq:gini_circ} and \eqref{eq:circular}. 
This is why the strong convergence of GUEs can easily be extended to self-adjoint polynomials in GEs.

\begin{theorem}\label{theorem:clt}
Suppose we have $k$ independent $n \times n$ GUEs or GEs $X_1, \ldots, X_k$.
Then, for the sequence of random matrices $\{W_{k,n}\}_{n=1}^\infty$ defined by
\eq{\label{eq:W_sum}
W_{k,n} = \sum_{i=1}^k X_i^* X_i \ ,
 }
almost surely we have
\eq{ \label{eq:asymptotic_bound}
\lim_{n \to \infty} \mu_{\min} (W_{k,n})= \left(\sqrt{k} -1\right)^2   
\eqtext{and} 
 \lim_{n \to \infty} \mu_{\max} (W_{k,n}) = \left(\sqrt{k} +1\right)^2  
 \ ,
} 
where $\mu_{\min}(\cdot)$ and $\mu_{\max}(\cdot)$ are respectively the minimum and maximum eigenvalues.
\end{theorem}
\begin{proof} 
Let $\{x_i\}_{i=1}^k$ be a semi-circular or circular system, and $w_k:=x_1^*x_1+x_2^*x_2 + \cdots + x_k^*x_k$.  
In view of the strong convergence in Theorem \ref{thm:strong}, it suffices to compute the infimum and the supremum of the support of $\mu_{w_k}$. Recall that each $x_ix_i^*$ has the Marchenko-Pastur distribution with rate 1. According to \eqref{eq:R}, $w_k$ has the R-transform $R_{w_k}(z)= k/(1-z)$, so that it has the Marchenko-Pastur distribution with rate $k$. In particular, the infimum and the supremum of the support of $\mu_{w_k}$ are respectively given by $k+1 - 2\sqrt{k}$ and $k+1 + 2\sqrt{k} $. 
\end{proof}

We will also use the norm of a general homogeneous quadratic polynomial in GUEs. The following function will be crucial in the analysis of our random CP map \eqref{eq:our-map}, for example see \eqref{eq:key}.
However, our statement is more general than needed as we apply the following proposition only to self-adjoint polynomials.

\begin{proposition}\label{proposition:f}
For a semi-circular system $\{s_1, \ldots, s_k\}$,  a circular system $\{\ce_1, \ldots, \ce_k\}$ and a matrix $A \in M_k(\mathbb C)$, 
we have 
\eq{\label{eq:f_def}
\left \| \sum_{i,j =1}^k  a_{i,j} \, s_i s_j \right\| 
=
\left \| \sum_{i,j =1}^k  a_{i,j} \, \ce_i^* \ce_j \right\| 
\ ,
}
where $a_{i,j} = (A)_{i,j}$ are the $(i,j)$-elements of $A$.
The common value above will be denoted by $f(A)$ below.
\end{proposition}
\begin{proof}
First, we define
\eq{
 s_A=\sum_{i,j =1}^k  a_{i,j} \, s_i s_j \qquad \text{and}\qquad \ce_A =\sum_{i,j =1}^k  a_{i,j} \, \ce_i^* \ce_j  \ .
}
By the $C^*$-condition, it suffices to prove that $\|s_A^*s_A\|=\|\ce_A^* \ce_A\|$, which follows if the distributions of $s_A^*s_A$ and $\ce_A^* \ce_A$ coincide since the state is faithful. The element $\ce_A^* \ce_A$ has the obvious formula
\[
\ce_A^* \ce_A = \sum_{i,j, x, y=1}^k \overline{a_{j,i}} a_{x,y} \ce_i^* \ce_j \ce_x^* \ce_y.
\]
With the notation 
\[
\alpha_{i_1,i_2,\dots,i_{4r}} =\overline{a_{i_2,i_1}} a_{i_3,i_4}\overline{a_{i_6,i_5}}a_{i_7,i_8} \cdots \overline{a_{i_{4r-2},i_{4r-3}}}  a_{i_{4r-1},i_{4r}},
\]
the free cumulants of $\ce_A^* \ce_A$ are computed as 
\eq{
\kappa_r[\ce_A^* \ce_A] 
&= \sum_{i_1,i_2,\dots, i_{4r}\in[k]} \alpha_{i_1,i_2,\dots, i_{4r}} \kappa_r[\ce_{i_1}^* \ce_{i_2} \ce_{i_3}^* \ce_{i_4}, \ce_{i_5}^* \ce_{i_6} \ce_{i_7}^* \ce_{i_8}, \dots, \ce_{i_{4r-3}}^* \ce_{i_{4r-2}} \ce_{i_{4r-1}}^* \ce_{i_{4r}}] \\
&= \sum_{i_1,i_2,\dots, i_{4r}\in [k]} \alpha_{i_1,i_2,\dots, i_{4r}} \sum_{\substack{\pi \in NC(4r) \\ \pi\lor\hat{1}^r_4=\hat{1}_{4r}}} \kappa_\pi[\ce_{i_1}^*, \ce_{i_2}, \ce_{i_3}^*, \ce_{i_4}, \ce_{i_5}^*,  \dots, \ce_{i_{4r-1}}^*, \ce_{i_{4r}}] \\
&= \sum_{i_1,i_2,\dots, i_{4r}\in [k]} \alpha_{i_1,i_2,\dots, i_{4r}} \sum_{\substack{\pi \in NC_2(4r) \\ \pi\lor\hat{1}^r_4=\hat{1}_{4r}}} \kappa_\pi[\ce_{i_1}^*, \ce_{i_2}, \ce_{i_3}^*, \ce_{i_4}, \ce_{i_5}^*,  \dots, \ce_{i_{4r-1}}^*, \ce_{i_{4r}}] \ ,
}
where \cite[Theorem 11.12]{nica2006lectures} was used on the second line and the known fact
\eq{\label{eq:2ndm2}
\kappa_{2}[\ce_i, \ce_j^*] = \kappa_2[\ce_i^*, \ce_j] = \delta_{i,j}, 
\text{ and the other free cumulants are zero,}
}
was used on the third line. 

Note that for $\pi \in NC_2(4r)$, each block of $\pi$ consists of an odd number and an even number; otherwise there would be a crossing. This fact together with \eqref{eq:2ndm2} implies that for such $\pi$
\[
\kappa_\pi[\ce_{i_1}^*, \ce_{i_2}, \ce_{i_3}^*, \ce_{i_4}, \ce_{i_5}^*,  \dots, \ce_{i_{4r-1}}^*, \ce_{i_{4r}}] = \prod_{\{p,q\} \in \pi}\delta_{i_p,i_q} \ . 
\]
We thus arrive at the expression 
\[
\kappa_r[\ce_A^* \ce_A] = \sum_{\substack{\pi \in NC_2(4r) \\ \pi\lor\hat{1}^r_4=\hat{1}_{4r}}} \sum_{\substack{i_1,i_2,\dots, i_{4r}\in [k] \\ i_p=i_q \text{~if~} \{p,q\} \in \pi}} \alpha_{i_1,i_2,\dots, i_{4r}} \ .
\]
A similar (actually simpler) reasoning yields the same formula for $\kappa_r[s_A^*s_A]$, where the fact \eqref{eq:2ndm2} needs to be replaced by 
\eq{\label{eq:2ndm}
\kappa_{2}[s_i, s_j] = \delta_{i,j} \ , 
\text{ and the other free cumulants are zero.}
}
The details are omitted,
but notice that one can pair only elements at odd and even positions, giving the same combinatorial structure as before.
\end{proof}

\begin{remark} The proof above is actually applicable to a more general result: for every 
polynomial in $\C^*[z_1,\dots,z_k]$ of the following form
\eq{
P(z_1,z_2,\dots, z_k) = \alpha + \sum_{p=1}^q \sum_{i_1,i_2,\dots, i_{2p} \in[k]} \alpha_{i_1,\dots, i_{2p}} z_{i_1}^* z_{i_2} z_{i_3}^* z_{i_4} \cdots z_{i_{2p-1}}^* z_{i_{2p}} \ , 
}
$P(s_1,s_2,\dots,s_k)$ and $P(\ce_1,\ce_2,\dots,\ce_k)$ have the same $*$-distribution. 
\end{remark}

\begin{theorem}\label{theorem:bounding-condition}
For $f(\cdot)$ in Proposition \ref{proposition:f} and
a positive semi-definite matrix $A \in M_k(\mathbb C)$ we have
\eq{\label{eq:f_estimate}
f(A)= \min_{x \in (0,1/\|A\|)} h(x,A) \ , \quad \text{where} \quad
h(x,A) = 
\frac{1}{x}+\sum_{i=1}^k\frac{\lambda_i}{1-\lambda_i x} \ ,
}
and $\lambda_i$'s are the eigenvalues of $A$. 
Moreover, $x\mapsto h(x,A)$ is strictly convex on $(0,1/\|A\|)$, and the above minimum corresponds to a unique critical point of $h(x,A)$.
In particular,  $f(\cdot)$ is rotationally invariant on the set of positive semi-definite matrices. 
\end{theorem}
\begin{proof}
We follow the notation in the proof of Proposition \ref{proposition:f}. It suffices to focus on a semi-circular system. 

Since $A$ is positive semi-definite, $s_A$ is a positive element, so that $f(A)= \sup \text{Spec}(s_A)$. 
The free cumulants of $s_A$ will be
\eq{
\kappa_r[s_A] = \sum_{\substack{i_1,i_2,\cdots,i_{r}\in [k]\\j_1,j_2,\cdots,j_{r}\in [k]}}
a_{i_1,j_1} \cdots a_{i_r,j_r}\,\, \kappa_{r}[s_{i_1}s_{j_1},
\ldots, s_{i_r}s_{j_r}] \ .
}
Now, by using \cite[Theorem 11.12]{nica2006lectures}, we have
\eq{\label{eq:no1}
\kappa_r[s_{i_1}s_{j_1},\ldots, s_{i_r}s_{j_r}] 
= \sum_{\substack{\pi\in NC_2(2r)\\\pi\lor \hat{1}^r_2=\hat{1}_{2r}}}\kappa_\pi[s_{i_1}, s_{j_1},\cdots,s_{i_r}, s_{j_r}] 
=\kappa_\nu [s_{i_1}, s_{j_1},\cdots,s_{i_r}, s_{j_r}] \ .
}
Here we used the fact \eqref{eq:2ndm} and set $\nu=\{\{2r,1\},\{2,3\},\{4,5\}, \ldots, \{2r-2,2r-1\}\}$. 
Hence
\eqind{\label{eq:p-norm-A}
\kappa_r[s_A]
=\sum_{i_1,i_2,\cdots,i_{r}\in [n]}
a_{i_1,i_2}a_{i_2,i_3} \cdots a_{i_r,i_1}
=\trace [A^r]
=\sum_{i=1}^k\lambda_i^r
}
where $\lambda_1,\ldots,\lambda_k$ are the eigenvalues of the  matrix $A$.

Then we get the R-transform of $\mu_{s_A}$: 
\eqind{
R_{s_A}(z)
=\sum_{r=0}^\infty\lrb\sum_{i=1}^k\lambda_i^{r+1}\rrb z^r
=\sum_{i=1}^k\lrb\sum_{r=0}^\infty\lambda_i^{r+1} z^r\rrb
=\sum_{i=1}^k\frac{\lambda_i}{1-\lambda_iz}
}
which is well-defined for $|z| < 1/ \|A\|$.
Hence, similarly as before, 
\eq{
G_{s_A}^{-1}(z)=\frac{1}{z}+\sum_{i=1}^k\frac{\lambda_i}{1-\lambda_iz} =: h(z,A)  \ .
}

Finally, to show \eqref{eq:f_estimate} one needs to know where analyticity of $G_{s_A}(z)$ breaks down
starting from $+ \infty$, which indicates the upper edge of the support of the measure $\mu_{s_A}$.
Detailed explanations are found, for example, in \cite{mingo2017free}.
Now, since $\displaystyle \frac{1}{z}$ and $\displaystyle \frac{\lambda_i}{1-\lambda_i z}$'s are all strictly convex on $(0,1/\|A\|)$, so is their sum. Then, the fact that $\displaystyle \lim_{z\searrow 0}h(x,A) = +\infty = \lim_{z \nearrow 1/\|A\|} h(z,A)$ 
shows that $h(z,A)$ is strictly convex  on $(0, 1/\| A\|)$ and has the unique minimum value $v$ at $c \in (0, 1/ \| A\|)$. 
This means that $h(z,A)$ is strictly decreasing on $(0, c)$ but not on $(0,d)$ with $d>c$. Hence,
$G_{s_A}(z)$ is analytically continued to $\mathbb C \setminus [0,v]$, but not to $\mathbb C \setminus [0,u]$ for any $u<v$. 
By \eqref{eq:support}, this completes the proof.
\end{proof}

\begin{remark}
The result in \eqref{eq:p-norm-A} had been shown for real symmetric matrices
by a different proof in \cite{hiwatashi1999free}. In addition, Theorem \ref{theorem:bounding-condition} is included in \cite{sato_thesis}.
\end{remark}

Before concluding this section, let us introduce a lemma for later use. 
\begin{lemma}\label{lemma:4th-distribution}
For a system of semi-circular or circular elements $\{x_i: i \in I\}$, we have
\eq{
\varphi(x_i x_j^* x_u x_v^*) = 
\begin{cases}
2& i=j=u=v\\
1& i=j \not=u=v \text{ or }i=v \not = j=u  \\
0& \text{otherwise}
\end{cases}
}
\end{lemma}
\begin{proof}
By using \eqref{eq:2ndm} and \eqref{eq:2ndm2}, non-vanishing terms correspond to the following two non-crossing set partitions:
$\{\{1,2\},\{3,4\}\}$ and $\{\{1,4\},\{2,3\}\}$. Therefore, we have
\eq{
\varphi(x_i x_j^* x_u x_v^*) = \kappa_2[x_i, x_j^*] \kappa_2[x_u, x_v^*] + \kappa_2[x_v^*, x_i] \kappa_2[x_j^*, x_u]
= \delta_{i,j}\delta_{u,v} + \delta_{i,v} \delta_{j,u} \ .
}
This completes the proof. 
\end{proof}

\subsection{Haagerup inequality for semi-circular and circular systems}\label{sec:haagerup}

The quantity $f(\cdot)$ in \eqref{eq:f_def} has an explicit form only in exceptional cases, e.g.\ $A=I_k$ as in Theorem \ref{theorem:clt}. For a general matrix $A$, a Haagerup-type inequality provides a bound for $f(\cdot)$ from above, which is stated in Theorem \ref{thm:Haagerup} below. What we need is already stated in \cite{bozejko1991q} but 
 we will give a self-contained proof for a more precise constant; see Remark \ref{remark:const}. 
Also, Theorem \ref{thm:Haagerup} shares the same spirit as \cite[Theorem 5.3.4]{biane1998stochastic}, in which diagonal elements are all zero.

Since the expression \eqref{eq:f_estimate} for $f(\cdot)$ does not depend on a choice of a $C^*$-probability space $(\mathcal A, \varphi)$, we may and do take a specific model in which semi-circular elements have useful decompositions. Let $F(H)$ be the full Fock space over a Hilbert space $H$ of dimension $k$ with a fixed orthonormal basis $\{e_i\}_{i=1}^k$, 
\begin{equation}
F(H) = \C \Omega \oplus \bigoplus_{n=1}^\infty H^{\otimes n},
\end{equation}
where $\Omega$ is a unit vector.  Let $\varphi_\Omega$ be the vacuum state on $B(F(H))$, that is, $\varphi_\Omega(a) = \langle \Omega| a | \Omega \rangle $ for $a \in B(F(H))$. Let $\ell^*(f)$ be the left creation operator for $f \in H$, namely, it acts on $F(H)$ adding the new tensor component $f$ from the left, and let $\ell(f)$ be its adjoint operator. Define $\ell_i := \ell(e_i)$ and $\ell_i^*:= \ell^\ast(e_i)$ for $i \in [k]$. 
 Then the family $\{s_i\}_{i=1}^k$ defined by 
\begin{equation}
s_i = \ell_i + \ell_i^* 
\end{equation}
is known to be a semi-circular system. 
The state $\varphi_\Omega$ is not faithful on $B(F(H))$, but is faithful on $\mathcal A_k := C^*\{s_i: 1 \le i \le k\}$, the $C^*$-algebra generated by $\{s_i\}_{i=1}^k$. For further details, see \cite[Lecture 7]{nica2006lectures}. 

Now, we state the main estimate in this subsection. 
\begin{theorem} \label{thm:Haagerup} Let $k \in \N$ and $A \in M_k(\C)$.  
The following bound holds. 
\eq{ \label{eq:haagerup1}
f(A) \leq 3\|A\|_2 + |\Tr(A)| \ .
}
In particular, if $\trace[A]=0$ then we have
$f(A) \leq 3 \| A \|_2.$
\end{theorem}

\begin{lemma}\label{lem:Haagerup} For any $k \in \N$ and $A =(a_{i,j}) \in M_k(\C)$ the following equalities hold.  
\[
\left\|\sum_{i,j=1}^k a_{i,j} \ell_i \ell_j\right\| = \|A\|_2  \qquad\text{and}\qquad \left\|\sum_{i,j=1}^k a_{i,j} \ell_i^* \ell_j\right\| = \|A\|_\infty. 
 \]
\end{lemma}

\begin{proof} Those inequalities are found in \cite[Theorem (I), p.\ 320]{bozejko1991q} and \cite[Lemma 1, p.\ 318]{bozejko1991q}.  For the sake of convenience, they are reproved below. 

For the first formula, using the $C^*$-condition and the property $\ell_i \ell_j^* = \delta_{i,j}\mathbf1_{\mathcal A_k}$ one has 
\begin{align*}
\left\|\sum_{i,j=1}^k a_{i,j} \ell_i \ell_j\right\|^2 
&= \left\|\left(\sum_{i,j=1}^k a_{i,j} \ell_i \ell_j \right)\left(\sum_{p,q=1}^k a_{p,q} \ell_p \ell_q \right)^*\right\| \\
&= \left\|\sum_{i,j,p,q=1}^k a_{i,j}\overline{a_{p,q}} \ell_i \ell_j \ell_q^* \ell_p^*\right\|\\
&=   \left\|\sum_{i,j=1}^k a_{i,j}\overline{a_{i,j}} \mathbf1_{\mathcal A_k}\right\| = \|A\|_2^2. 
\end{align*}

For the second formula, define a linear map $\alpha\colon M_n(\C) \to B (F(H))$ by 
\begin{equation*}
\alpha(A) = \sum_{i,j=1}^k a_{i,j} \ell_i^* \ell_j. 
\end{equation*}
This map is an injective $\ast$-homomorphism between $C^*$-algebras, and hence is an isometry. 
\end{proof}

\begin{proof}[Proof of Theorem \ref{thm:Haagerup}]
Let $a_{i,j}$ be the $(i,j)$-entry of $A$. By the triangle inequality, Lemma \ref{lem:Haagerup},  $\ell_i \ell_j^* = \delta_{i,j}\mathbf1_{\mathcal A_k}$ and the well known inequality $\|A\|_\infty \le \|A\|_2$ we obtain
\begin{align*}
\left\|\sum_{i,j=1}^k a_{i,j} s_i s_j\right\| 
&= \left\|\sum_{i,j=1}^k a_{i,j} (\ell_i + \ell_i^*)(\ell_j + \ell_j^*)\right\| \\
& \le \left\|\sum_{i,j=1}^k a_{i,j} \ell_i\ell_j\right\| + \left\|\sum_{i,j=1}^k a_{i,j} \ell_i\ell_j^*\right\| + \left\|\sum_{i,j=1}^k a_{i,j} \ell_i^*\ell_j\right\| + \left\|\sum_{i,j=1}^k a_{i,j} \ell_i^*\ell_j^*\right\|  \\
&= \|A\|_2  + |\Tr(A)| + \|A\|_\infty + \|A\|_2\\
&\le 3\|A\|_2 + |\Tr(A)|. 
\end{align*}
\end{proof}

\begin{remark}\label{remark:const}
The Cauchy-Schwarz inequality $|\Tr(A)| \le \sqrt{k} \|A\|_2$ and \eqref{eq:haagerup1} give
\eq{\label{eq:inequality_f}
f(A) \leq (3+\sqrt{k})\|A\|_2 \ . 
}
The constant on the RHS of \eqref{eq:inequality_f} contains $\sqrt{k}$, which depends on the size of the matrix $A$, while \cite[Theorem (III), p.\ 320]{bozejko1991q} gives a similar inequality whose constant does not depend on $k$. We think that $\sqrt{k}$ is necessary, since for $A=I_k$ the operator $s_1^2 + s_2^2 + \cdots + s_k^2$ follows the Marchenko-Pastur law with parameter $k$, and hence its norm is $(1+\sqrt{k})^2$,  and $\|I_k\|_2 = \sqrt{k}$.
\end{remark}

\section{Violation of multiplicativity}\label{sec:vio-cp}

\subsection{Nearly trace-preserving} 

The following theorem states that, asymptotically as $n\to\infty$, the random CP maps in \eqref{eq:our-map} are nearly trace-preserving for $k \gg 1$.

\begin{theorem}\label{theorem:almost-trace-preserving2}
For every $k \in \N$ the following statements are almost surely true as $n \to \infty$.
\eq{\left(1- \frac{1}{\sqrt{k}} \right)^2 
\leq
\liminf_{n \to \infty} \min_{\rho \in D_n} \trace \left[ \Phi_n(\rho) \right]  
\leq
\limsup_{n \to \infty} \max_{\rho \in D_n} \trace \left[ \Phi_n(\rho) \right]  
\leq 
\left(1+ \frac{1}{\sqrt{k}} \right)^2  
\ .
}
The same bounds hold for $\Phi_n^c$ (see \eqref{eq:complement} for the definition).
\end{theorem}
\begin{proof}
Notice that by linearity the optimum is achieved when $\rho$ is a rank-one projection, which we denote by $|x \rangle \langle x|$. Then, we calculate:
\eq{
\trace[\Phi_n(|x \rangle \langle x|)] =  \frac 1 k  \langle x| \left [ \sum_{i=1}^k X_iX_i^*  \right] |x\rangle
}
and this means that 
\eq{
\min_{\rho \in D_n} \trace \left[ \Phi_n(\rho) \right]  = \mu_{\min}(W_{k,n})/k
\quad\text{and}\quad
\max_{\rho \in D_n} \trace \left[ \Phi_n(\rho) \right]  = \mu_{\max}(W_{k,n})/k \ ,
}
where $\mu_{\min}(W_{k,n})$ and $\mu_{\max}(W_{k,n})$ are the smallest and the largest eigenvalues of $W_{k,n}$ in \eqref{eq:W_sum}.
Hence, Theorem \ref{theorem:clt} proves our claim.
Also, it is clear that $\trace \left[ \Phi_n(\rho) \right] = \trace \left[ \Phi_n^c(\rho) \right]$ from the definition of complementary channel in \eqref{eq:complement}.
\end{proof}

\subsection{Convergence of maximum output norms}
In this section, we show that the MO$p$N of the random CP maps in \eqref{eq:our-map}
converge almost surely to some value formulated by the function in \eqref{eq:f_estimate}.
Similar free-probability methods had been adopted in \cite{belinschi2016almost,  cfn3,collins2015estimates}
to show the additivity/multiplicativity violation. 
Below, most of claims on $\Phi_n$ and $\Phi_n^c$ are interchangeable because of Proposition \ref{proposition:complement}, and this fact will be repeatedly used without mentioning.
\begin{definition}
We define the following compact set: for $q>1$
\eq{\label{eq]compact_domain}
\hat D_{k,q} = \{A \in M_k (\mathbb C): A \geq 0, \, \| A \|_q = 1 \} \ .
}
Also, for a matrix $A \in M_k(\C)$, a sequence of $k$ independent $n \times n$ GUEs $\{(S_1,\ldots,S_k)\}_{n=1}^\infty$ and GEs  $\{(\gini_1,\ldots,\gini_k)\}_{n=1}^\infty$ we define 
\eq{\label{eq:f_n}
f_n^S(A) =  \left \| \sum_{i,j =1}^k  a_{i,j} \, S_i S_j \right\| 
\eqtext{and} 
f_n^\gini(A) =  \left \| \sum_{i,j =1}^k  a_{i,j} \, \gini_i^* \gini_j \right\| 
\ ,
}
where $a_{i,j} = (A)_{i,j}$ are the $(i,j)$-elements of $A$.
When it is not necessary to distinguish $f_n^S(\cdot)$ and $f_n^\gini(\cdot)$, we denote both by $f_n(\cdot)$.

Note that $\{(S_1,\ldots,S_k)\}_{n=1}^\infty$ or $\{(\gini_1,\ldots,\gini_k)\}_{n=1}^\infty$ defines a sequence of random CP maps $\{\Phi_n\}_{n=1}^\infty$ via \eqref{eq:our-map}, and a sequence of random variables
$f_n(A)$.
\end{definition}

Before stating lemmas and theorems we introduce a key idea in this subsection:
\eq{\label{eq:key}
\max_{|x \rangle \langle x| \in D_n}\trace \left[ \Phi_n^c (|x \rangle \langle x| ) A \right]
&=\max_{|x \rangle \langle x| \in D_n} \frac 1k  \sum_{i,j} \trace \left[X_i |x \rangle \langle x | X_j^* \right]  \trace \left[ |i \rangle \langle j| A \right] \\
&= \frac 1k  \max_{|x \rangle \langle x| \in D_n}  \langle x | \left(  \sum_{i,j} a_{j,i} X_j^* X_i  \right)  |x \rangle 
= \frac 1k \, f_n(A) \ ,
} 
which will be used later, for example:
\begin{lemma}\label{lemma:uniform}
Fix $1<q$, and almost surely, $f_n$ in \eqref{eq:f_n} converges uniformly on $\hat D_{k,q}$ to $f$ in \eqref{eq:f_def}:
\eq{
\lim_{n \to \infty} f_n (A) = f(A) \quad\text{for}\quad A \in \hat D_{k,q} \ .
}
\end{lemma}
\begin{proof}
First, we form a countable dense subset of $\hat D_{k,q}$, which we denote by $N_{k,q}$. Then, almost surely, the sequence $\{(X_1,\ldots,X_k)\}_{n=1}^\infty$ have the following properties via strong convergence.
\eq{
\forall A \in N_{k,q}, \quad \lim_{n \to \infty} f_n(A) = f(A), \quad
\text{and} \quad
\forall i \in [k], \quad  \limsup_{n \to \infty} \| X_i \|_\infty \leq 2 \ .
}
Now, fix $\epsilon >0$ 
and, by compactness, there is a finite subset  $N_{k,q,\epsilon} \subseteq N_{k,q}$ to form an open cover: 
\eq{
\hat D_{k,q} \subseteq \bigcup_{A \in N_{k,q,\epsilon}} B(A, \epsilon) \ . 
}
Here, $B(A, \epsilon)$ is an open ball with center $A$ and the radius $\epsilon >0$ with respect to the Euclidean distance. 
Then, $\exists N_0 \in \mathbb N$, $\forall n \geq N_0$,
\eq{
\forall A^\prime \in N_{k,q,\epsilon}, \quad |f_n(A^\prime) - f(A^\prime)| < \epsilon, \qquad
\text{and} \quad
\forall i \in [k],\quad \| X_i \|_\infty \leq 3 \ .
}
In this case, for all $A \in \hat D_{k,p}$ and $n \geq N_0$ we can choose a proper corresponding $A^\prime \in N_{k,q,\epsilon}$ so that
\eq{
\left|f_n(A) - f(A) \right|
\leq
\left|f_n(A) - f_n(A^\prime) \right|+
\left|f_n(A^\prime) - f(A^\prime) \right|+
\left|f(A^\prime) - f(A) \right| 
\leq 9 k \epsilon + \epsilon + 4 k \epsilon \ .
}
Since the RHS does not depend on the choice of $A \in \hat D_{k,p}$,
the claim has been proved. 
\end{proof}

Now, we discuss the convergence of MO$p$N for $p>1$:
\begin{theorem}[Convergence of MO$p$N] \label{theorem:p-norm-convergence}
For all $p>1$, almost surely we have
\eq{\label{eq:max-f}
\lim_{n \to \infty } \| \Phi_n\|_{p} = \lim_{n \to \infty } \| \Phi_n^c\|_{p} =  \frac 1k  \max_{A \in \hat D_{k,q}}   f(A) \ ,
}
where  $q>1$ is such that $1/p +1/q =1$, and
the above function $f(\cdot)$ is defined in \eqref{eq:f_def}.
\end{theorem}
\begin{proof}
Fix $p>1$ and calculate MO$p$N, by restricting inputs to rank-one projections $|x \rangle\langle x|$ as before:
\eq{
\|\Phi_n^c \|_p 
= \max_{|x \rangle \langle x| \in  D_n} \| \Phi_n^c (|x \rangle \langle x| ) \|_p
= \max_{|x \rangle \langle x|  \in  D_n}  \max_{A \in \hat D_{k,q}} \trace [\Phi_n^c(|x \rangle \langle x|) A] 
=\frac{1}{k}\max_{A \in \hat D_{k,q}}f_n(A)}
\ .
Here, we used Lemma \ref{lemma:duality} and \eqref{eq:key}.

Then, by Lemma \ref{lemma:uniform}, we have almost surely as $n \to \infty$
\eq{
 \max_{A \in \hat D_{k,q}}f_n(A)
 \quad \longrightarrow \quad \max_{A \in \hat D_{k,q}} f(A) \ .
}
This completes the proof.
\end{proof}

The above convergence was discussed in terms of convergence of random sets in \cite{cfn3}, i.e.
the random sets of output states converge almost surely to some convex set,
but we do not need a corresponding result in this paper. 

To conclude this section, we also discuss the convergence of the maximum output $\infty$-norm:
\begin{theorem} [Convergence of MO$\infty$N] \label{theorem:largest-eigenvalue-convergence}
Almost surely, we have
\eq{
\lim_{n \to \infty} \| \Phi_n \|_\infty = \lim_{n \to \infty} \| \Phi_n^c \|_\infty = \frac 4k \ .
}
\end{theorem}
\begin{proof}
First, the convexity of norms enables inputs to be restricted to rank-one projections, which we denote by $|x \rangle\langle x|$. 
Then, similarly as in 
 \eqref{eq:key},
\eq{
\|\Phi_n^c \|_\infty
= \max_{|x \rangle \langle x| \in  D_n} \| \Phi_n^c (|x \rangle \langle x| ) \|_\infty
= \max_{|x \rangle \langle x|  \in  D_n}  \max_{A } \trace [\Phi_n^c(|x \rangle \langle x|) A] 
= \frac{1}{k}\max_{A} f_n(A) \ ,
}
where $A$s are rank-one projections.

Next, we apply Theorem \ref{lemma:uniform} and Theorem \ref{theorem:bounding-condition} to have 
\eq{
\lim_{n \to \infty} \| \Phi_n^c \|_\infty  = \frac 1k \min \left\{\frac{1}{z}+ \frac{1}{1- z}: 0 < z <1\right\} =\frac 4 k \ .
}
The calculation was completed. 
\end{proof}

\subsection{Violation of multiplicativity}\label{sec:vilation_phi}
First, we find an upper bound for the single channel:
\begin{theorem}\label{theorem:bound-p-single}
For all $p>1$ choose $k$ large enough, and then almost surely 
\eq{
\lim_{n \to \infty } \| \Phi_n\|_{p} = \lim_{n \to \infty } \| \Phi_n^c\|_{p} \leq  \left(\left(\frac 4 k \right)^p +\left(\frac1{k-1} \right)^{p-1} \left[ 1-\frac 3k + \frac 2 {\sqrt{k}} \right]^p \right)^{\frac 1p} \ .
}
\end{theorem}
\begin{proof}
Firstly, fix a sample sequence $\{\Phi_n\}_{n=1}^\infty$ which shows uniform convergence as in Lemma \ref{lemma:uniform} and Theorem \ref{theorem:p-norm-convergence}. Then, choose $|\tilde x_n \rangle \langle \tilde x_n| \in  D_n$ and $\tilde A_n \in \hat D_{k,q}$ properly, in terms of Lemma \ref{lemma:duality}, so that 
\eq{\label{eq:fn}
\|\Phi_n^c \|_p = \|\Phi_n^c(|\tilde x_n \rangle \langle \tilde x_n|)\|_p
=\trace [\Phi_n^c(|\tilde x_n \rangle \langle \tilde x_n|) \tilde A_n] 
\leq \frac 1k  f_n(\tilde A_n)
\leq \frac{1}{k} \max_{A \in \hat D_{k,q}}
f_n(A) 
= \|\Phi_n^c \|_p \ ,
}
where \eqref{eq:key} was used again.
This means that 
\eq{\label{eq:optimal_matrices_convergence}
 \lim_{n \to\infty} f_n(\tilde A_n)
=  f(\tilde A)
=k\lim_{n \to \infty}\trace [\tilde B_n \tilde A_n] \ ,
}
where we wrote $\tilde B_n = \Phi_n^c(|\tilde x_n \rangle \langle \tilde x_n|)$ and chose an arbitrary $\tilde A$ from the following set of optimal matrices:
\eq{
M_{k,q} = \left\{\tilde A \in \hat D_{k,q}: f(\tilde A) =  \max_{A \in \hat D_{k,q}} f(A) \right\} \ .
}

Secondly, define the $\epsilon$-neighborhood of $M_{k,q}$:
\eq{
M_{k,q,\epsilon} = \left\{A \in \hat D_{k,q}: d(A, M_{k,q})<\epsilon \right\} \ ,
}
where $d(\cdot,\cdot)$ is made of the Euclidean norm.
Then, we claim that $\forall \epsilon > 0$, $\exists N_0 \in \mathbb N$, $\forall n \geq N_0$, it holds that $\tilde A_n \in M_{k,q,\epsilon}$.
Indeed, otherwise, there is some $\epsilon_0>0$ and a sub-sequence $\displaystyle \{\tilde A_{n_m}\}_{m=1}^\infty$ such that $\tilde A_{n_m} \not \in M_{k,q,\epsilon_0}$.
Since $\hat D_{k,q} \setminus M_{k,q,\epsilon_0}$ is closed and bounded, there is a convergent sub-sequence $\displaystyle \{\tilde A_{n_{m_\ell}}\}_{\ell=1}^\infty$ such that $\displaystyle \lim_{\ell \to \infty} \tilde A_{n_{m_\ell}} = A_0$ with $d(A_0, M_{k,q}) \geq \epsilon_0$. On the other hand, suppose $\|\tilde A_{n_{m_\ell}} - A_0\|_2 < \epsilon_1$, and
 uniform convergence of $f_n(\cdot)$ in Lemma  \ref{lemma:uniform} indicates that
\eq{
\left| f_{n_{m_\ell}} (\tilde A_{n_{m_\ell}}) - f(A_0)\right|
\leq \left| f_{n_{m_\ell}} (\tilde A_{n_{m_\ell}}) - f(\tilde A_{n_{m_\ell}})\right|+
\left| f(\tilde A_{n_{m_\ell}}) - f(A_0)\right|
\leq \epsilon_1 + 4k\epsilon_1 
}
for large enough $\ell$.
This would yield a contradiction:
\eq{
\lim_{\ell \to \infty}   f_{n_{m_\ell}} (\tilde A_{n_{m_\ell}})  = f(A_0)< f(\tilde A) \ .
}
The above deduction means that there exists a sequence $\{\tilde A^{(n)}\}_{n=1}^\infty$ in $M_{k,q}$ such that $\|\tilde A_n - \tilde A^{(n)}\|_2 \to 0$ as $n\to\infty$.

Thirdly, then note that via Lemma \ref{lemma:optimal-shape} all matrices $\tilde A^{(n)}$ have the following shape of eigenvalues:
$ (\alpha,\beta, \ldots, \beta)$
for some $0 < \beta < \alpha$ (possibly depending on $n$). Now, Weyl's perturbation theorem, see for example \cite[Corollary III.2.6]{bhatia2013matrix}, indicates that
\eq{
\max_{i\in [k]} \left|\lambda_i^\downarrow(\tilde A_n) - \lambda_i^\downarrow(\tilde A^{(n)})\right|
\leq \| \tilde A_n - \tilde A^{(n)} \|_\infty \ ,
}
where $\lambda_i^\downarrow (A)$ is the $i$-th largest eigenvalue of the Hermitian matrix $A$. So, we can assume that the eigenvalue distribution of $\tilde A_n$ is asymptotically the same as that of $\tilde A^{(n)}$ with arbitrary small errors.
This in turn implies by Lemma \ref{lemma:duality} that we can assume $\tilde B_n$ has the following eigenvalue distributions asymptotically:
$(s,t,t, \ldots, t)$
for some $0 < t < s$ (possibly depending on $n$) because the error of this approximation can be made arbitrary small. 
On the other hand, Theorem \ref{theorem:largest-eigenvalue-convergence} implies that the larger eigenvalue $s$ is asymptotically equal to $\frac 4k$.
Moreover, Theorem \ref{theorem:almost-trace-preserving2} forces the condition that $ 1 + \frac 1k - \frac 2 {\sqrt{k}} \leq s + (k-1)t \leq  1 + \frac 1k + \frac 2 {\sqrt{k}}$.
This implies that, in the limit, maximum output $p$-norm is bounded from above by the $p$-norm of 
\eq{
\left\{s, \underbrace{t \ldots, t}_{k-1} \right\} \eqtext{where} s = \frac 4 k
\eqtext{and hence} t = \frac 1{k-1} \left[ 1-\frac 3k + \frac 2 {\sqrt{k}} \right] \ ,
}
which corresponds to the worst possible case under the above necessary constraints.
\end{proof}

\begin{theorem}\label{theorem:bell-input}
Almost surely, we have
\eq{
\lim_{n \to \infty}  \left(\Phi_n^c \otimes \bar \Phi_n^c \right)(|b_n\rangle \langle b_n|) 
= \frac{1}{k^2} I_{k^2} + \frac{1}{k} |b_k \rangle\langle b_k|
}
and in particular, 
\eq{\label{eq:bell-input}
\liminf_{n \to \infty} \left\| \Phi_n\otimes \bar \Phi_n \right\|_p 
=\liminf_{n \to \infty} \left\| \Phi_n^c\otimes \bar \Phi_n^c \right\|_p
\geq \left(\left( \frac 1k + \frac 1{k^2} \right)^p + (k^2-1) \left( \frac 1 {k^2}\right)^p \right)^{\frac 1p} \ .
}
Here, \eq{\label{eq:bell_state}
|b_n \rangle = \frac{1}{\sqrt{n}}\sum_{i = 1}^n  |i\rangle \otimes | i\rangle
}
is a Bell state, where $\{|i\rangle\}_i$ is the canonical basis in $\mathbb C^n$.
\end{theorem}

\begin{proof}
First, notice that 
\eq{
 \left(\Phi_n^c \otimes \bar \Phi_n^c \right) (|b_n \rangle \langle b_n|) 
&=\sum_{i,j,u,v =1}^k\frac1{k^2} \trace \left[(X_i \otimes \bar X_u)|b_n \rangle \langle b_n|
(X_j^* \otimes \bar X_v ^*) \right] \, |i,u\rangle \langle j,v|\\
&=\frac1{k^2}\sum_{i,j,u,v =1}^k \frac 1n \trace [X_i X_u^*  X_v X_j^*]
\, |i,u\rangle \langle j,v| \ .
}

Next, then by using the asymptotic freeness for GUEs and GEs, \cite{voiculescu1991} and Lemma \ref{lemma:4th-distribution} we calculate:
\eq{\label{eq:bell_limit}
 \left(\Phi_n^c \otimes \bar \Phi_n^c \right)  (|b_n \rangle \langle b_n|)
& \to \frac 1{k^2}\sum_{i,j,u,v =1}^k \varphi(x_ix_u^*x_vx_j^*) \, |i,u\rangle \langle j,v|  \qquad (n \to \infty)  \\
& = \frac 1{k^2} \left[ 
\sum_{\substack{i,u=1\\i \not =u}}^k | i,u \rangle \langle i,u|
+\sum_{\substack{i,j=1\\i \not =j}}^k | i,i \rangle \langle j,j |
+2\sum_{\substack{i=1}}^k | i,i \rangle \langle i,i |
\right] \\
&= \frac 1{k^2} I_{k^2} + \frac 1k |b_k \rangle \langle b_k| \ .
}
This showed our first claim. 
The bound in the second statement is the $p$-norm of this limit matrix.
\end{proof}

\begin{theorem}[Violation of multiplicativity]\label{theorem:violation-cp}
For $1.5<p\leq \infty$,
asymptotically as $n\to\infty$, the random CP maps in \eqref{eq:our-map} violate multiplicativity of MO$p$N for $k \gg 1$.
More precisely, choose large enough $k$, and then almost surely we have 
\eq{
\liminf_{n \to \infty}\| \Phi_n \otimes  \bar \Phi_n\|_p  > 
\lim_{n \to \infty}  
\| \Phi_n \|_p \cdot \| \bar \Phi_n \|_p \ ,
}
where $\bar \Phi$ is defined in \eqref{eq:our-map_conjugate}.
\end{theorem}
\begin{proof} 
The violation is obvious for $p=\infty$ from Theorem \ref{theorem:largest-eigenvalue-convergence} and Theorem \ref{theorem:bell-input}. Indeed,
\eq{
\liminf_{n \to \infty}\| \Phi_n \otimes  \bar \Phi_n\|_\infty
\geq \frac{1}{k} + \frac{1}{k^2}
> \frac{16}{k^2}
= \lim_{n \to \infty}  \| \Phi_n \|_\infty \cdot \| \bar \Phi_n \|_\infty \ ,
}
for $k \geq 16$, where the violation occurs.

For $1.5<p<\infty$ we
show the following inequality:
\eq{
k^{2p} \cdot \liminf_{n \to \infty}\| \Phi_n \otimes  \bar \Phi_n\|_p^p 
> k^{2p} \cdot \lim_{n \to \infty}  \| \Phi_n \|_p^{2p} \ .
}
To this end, we get the lower bound for the LHS from Theorem \ref{theorem:bell-input}:
\eq{\label{eq:lower_bound_pair_multi}
\left(k + 1\right)^p + (k^2-1) \approx  k^2 + k^p  \ .
}
On the other hand, the upper bound for the RHS can be obtained via Theorem \ref{theorem:bound-p-single}:
\eq{\label{eq:upper_bound_pair_multi}
&4^{2p} + k^2 \left( \frac{k}{k-1} \right)^{2p-2} \left[ 1- \frac{3}{k} + \frac{2}{\sqrt{k}}\right]^{2p} + 2 \cdot 4^p \cdot k \cdot \left(\frac{k}{k-1}\right)^{p-1}\left[ 1- \frac{3}{k} + \frac{2}{\sqrt{k}}\right]^{p}\\
&\approx k^2 \left[ \left( 1 + (2p-2)\frac{1}{k-1} \right)
\left( 1 + 2p \left(\frac{2}{\sqrt{k}} - \frac{3}{k} \right) \right) \right]
\approx k^2 + 4p \cdot k^{\frac{3}{2}} \ .
}
Therefore, if $p>1.5$, then $\eqref{eq:lower_bound_pair_multi} > \eqref{eq:upper_bound_pair_multi}$ for large enough $k$, so that
our claimed strict inequality holds. See Remark \ref{remark:degree} for the details.
\end{proof}

\begin{remark}\label{remark:degree}
In the proof of Theorem \ref{theorem:violation-cp}, we can see that the proven degree of violation of multiplicativity changes dramatically depending on the choice of $p$.
To this end, we compare \eqref{eq:lower_bound_pair_multi} and \eqref{eq:upper_bound_pair_multi}.
For $p>2$ the inequality holds as a result of difference in degrees of $k$. For $p=2$, however, it is the matter of the coefficients of $k^2$. In case $1.5<p<2$, the violation results only from the terms lower than $k^2$.
One can find in Appendix \ref{sec:<1.5} discussions on the case $1<p\leq 1.5$.

\end{remark}

To close this section, we write down a corollary derived from Theorem \ref{theorem:p-norm-convergence} and Lemma \ref{lemma:optimal-shape}.
The solution will lead us to the precise estimate of the limit MO$p$N, but at the moment we do not know how to calculate it. 
\begin{corollary}
For all $p>1$, and large enough $k$, almost surely we have 
\eq{
& \lim_{n \to \infty } \| \Phi_n\|_{p} = \lim_{n \to \infty } \| \Phi_n^c\|_{p} \\
& = \frac 1k  \max_{\alpha, \beta} \min_{0<x<\frac{1}{\alpha}} \left\{\frac{1}{x}+ \frac{\alpha}{1-\alpha x} + \frac{(k-1)\beta}{1-\beta x} : 0< \beta < \alpha, \quad
 \alpha^q + (k-1)\beta^q = 1\right\} \ .
}
\end{corollary}

\section{Violation of additivity of quantum channels}\label{sec:vio-qc}
\subsection{Rectifying CP maps to quantum channels}\label{sec:rectify}
In this section, we rectify the CP maps in \eqref{eq:our-map} into quantum channels.
First, we can define the rectifying operator $R$ based on Theorem \ref{theorem:clt}:
\begin{lemma}\label{lemma:rectifier}
For any $k \in \N (k\ge 2)$ and $0<\epsilon <1$, one can define the following operator almost surely for large enough $n$:
\eq{
R := \sqrt{k} \left(\sum_{i=1}^k X_i^*X_i \right)^{- \frac 12}
\qquad
\text{such that}
\qquad
 \frac{\sqrt{k(1 - \epsilon)}}{\sqrt{k} +1}  \leq R \leq  \frac{\sqrt{k(1 + \epsilon)}}{\sqrt{k} - 1} \ .
}
Here, $\{X_i\}_{i=1}^k$ are GUEs or GEs of size $n$.
\end{lemma}
\begin{proof}
Theorem \ref{theorem:clt} shows that $\displaystyle \sum_{i=1}^k  X_i^*X_i$ has
its positive inverse almost surely for large enough $n$.
Hence the asymptotic bounds in \eqref{eq:asymptotic_bound} show our assertion. 
\end{proof}
Therefore, we can define the following random quantum channels as long as $R$ is defined:
\eq{\label{eq:rectified-map}
\Psi_n (\rho) := \Phi_n(R \rho R) =  \frac 1k \sum_{i=1}^k X_i R \rho R X_i^* \ .
}
These are not only CP maps by the definition itself but also trace-preserving:
\begin{lemma}
The maps in \eqref{eq:rectified-map} preserve trace whenever they are defined.
\end{lemma}
\begin{proof}
The simple calculation
\eq{
\trace \left[ \Psi_n(\rho) \right] = \frac 1k\trace \left[ \rho R \sum_{i=1}^k X_i^*X_i R \right] = \trace [\rho]
}
proves our claim.
\end{proof}

The following lemma relates the quantum channels $\Psi_n$ to the  CP maps $\Phi_n$
in terms of MO$p$N.
\begin{lemma}\label{lemma:errors}
For any $0<\epsilon <1$, $k\in \mathbb N ~(k\ge2)$ and $m \in \mathbb N$, almost surely for large $n \in \mathbb N$,
we have the following bounds:
\eq{\label{eq:cp_qc}
\left(\frac{(1-\epsilon)k}{k+1+2\sqrt{k}}\right)^m  \left\| \bigotimes_{i=1}^m \Phi_n^{(i)} \right\|_p \leq \left\| \bigotimes_{i=1}^m \Psi_n^{(i)}  \right\|_p 
\leq \left(\frac{(1+\epsilon)k}{k+1-2\sqrt{k}}\right)^m \left\| \bigotimes_{i=1}^m \Phi_n^{(i)} \right\|_p
}
where $\Phi_n^{(i)} \in \{\Phi_n, \bar \Phi_n \}$  and $\Psi_n^{(i)} \in \{\Psi_n, \bar \Psi_n \}$.
\end{lemma}
\begin{proof}
As before 
we can restrict inputs to rank-one projections, which we denote by $|x \rangle\langle x|$, to calculate MO$p$N. Then, 
\eq{
\bigotimes_{i=1}^m \Psi_n^{(i)} (|x \rangle\langle x| )= \bigotimes_{i=1}^m \Phi_n^{(i)}  (R^{\otimes m} |x \rangle\langle x| R^{\otimes m}) 
}
but then $R^{\otimes m} |x \rangle\langle x| R^{\otimes m}$ is of rank one,
and its operator norm can be bounded from below and above by using Lemma \ref{lemma:rectifier}.
It implies the inequalities in \eqref{eq:cp_qc} because $p$-norms are homogeneous.
\end{proof}

We make a remark that the quantum channels $\Psi_n$ can be rewritten by truncated Haar-distributed unitary matrices $U \in \mathcal U (kn)$  when  $\{X_i\}_{i=1}^k$ are GEs, i.e.
\eq{
\Psi_n(\rho) = \trace_{\mathbb C^k} [V \rho V^*] \ ,
}
where for $1\leq i \leq kn$ and $1\leq j \leq n$,
\eq{
V =  (X_1^T, \ldots, X_k^T)^T \frac{R}{\sqrt{k}} 
\quad\text{ and } \quad V_{i,j} =  U_{i,j} \ .
}
Hence the proof about $\Psi_n$ in case of GEs provides a different proof method for the violations of the multipicativity and the additivity, compared with ones in \cite{hw_vio, fukuda2010entanglement, brandao2010hastings, aubrun2010nonadditivity, aubrun2011hastings, f-vio, belinschi2016almost}, where random quantum channels are generated by random isometies. 
The fact that $V$ is a truncation of a Haar unitary can be proved in the following way. Let $m=nk$ for notational simplicity. 

First, $V$ is unitarily left-invariant. 
Indeed,  let $X=(X_1^T, \ldots, X_k^T)^T$ and $W$ be an $m \times m$  unitary matrix. Since $X$ is a truncated Ginibre, the matrix $X' = W X$ has the same distribution as $X$. Also,  $R'= \sqrt{k}\left((X')^* X'\right)^{-1/2} = \sqrt{k}(X^*X)^{-1/2}  =R$. Hence $WV = X' \frac{R'}{\sqrt{k}}$ has the same distribution as $V$.  

Second, since $V$ is an isometry: $V^*V =I_{n}$, one can extend it to some $m \times m$ random unitary matrix $(V,V')$, which may not be Haar-distributed. Now, take a Haar unitary $U \in \mathcal U(m)$ independent of $(V,V')$, and let $(Y,Y'):=U(V,V')$.
Then $Y=UV$ has the same distribution as $V$, and $(Y,Y')$ is a Haar unitary because it is unitarily left-invariant.

\subsection{Additivity violation of quantum channels}\label{sec:violation}

The violation of multiplicativity proved in Theorem \ref{theorem:violation-cp} can be translated to the case of the rectified maps $\Psi_n$, which are now quantum channels.
\begin{theorem}[Violation of multiplicativity]\label{theorem:violation-cp2}
For $1.5 < p \leq \infty$, 
asymptotically as $n\to\infty$, the random quantum channels $\Psi_n$ in \eqref{eq:rectified-map} violate multiplicativity of MO$p$N for $k \gg 1$.
\end{theorem}
\begin{proof}
The violation in Theorem \ref{theorem:violation-cp} comes from the gap between \eqref{eq:upper_bound_pair_multi} and \eqref{eq:lower_bound_pair_multi}. 
Lemma \ref{lemma:errors} with $\epsilon=1/k$, for example, shows that rectifying $\Phi_n$ into $\Psi_n$ does not change significantly the gap and hence does not affect multiplicativity violation.
\end{proof}

For the additivity violation of MOE, the method of using the Haagerup inequality was employed in \cite{collins_random_unitary} for a Haar unitary quantum channel. A similar idea works for our rectified quantum channel $\Psi_n$.  

\begin{theorem}[Violation of additivity]\label{theorem:violation-qc}
Asymptotically as $n\to\infty$, the random quantum channels $\Psi_n$  in \eqref{eq:rectified-map} violate additivity of MOE for $k \gg 1$.
More precisely, almost surely, we have
\eq{
\limsup_{n \to \infty }S_{\min}(\Psi_n \otimes \bar \Psi_n) < \liminf_{n \to \infty} S_{\min}(\Psi_n) + \liminf_{n \to \infty}  S_{\min}(\bar \Psi_n)  \ .
}
Here, $\bar \Psi_n (\rho)$ is the complex conjugate of $\Psi_n (\rho)$:
\eq{\label{eq:rectified-map-c}
\bar \Psi_n (\rho) =  \frac 1k \sum_{i=1}^k  \overline{X_i R} \, \rho \, \overline{R X_i^*} \ .
}
\end{theorem}
\begin{proof}
As before, fix a sample sequence $\{\Phi_n\}_{n=1}^\infty$ (and hence $\{\Psi_n\}_{n=1}^\infty$) which shows uniform convergence as in Lemma \ref{lemma:uniform}. This is possible almost surely. 

Firstly, then we claim that almost surely as $n \to \infty$ 
\eq{\label{eq:2norm_bound}
\limsup_{n \to \infty} \max_{|x \rangle \langle x|   \in D_n}\|\Psi_n^c(|x \rangle \langle x|) - \tilde I_k \|_2 \leq \frac 3 {k+1-2\sqrt{k}} \ ,
}
where $\Psi_n^c$ is the complementary channel of $\Psi_n$; see Appendix \ref{appendix:complementary}.
Indeed, fix $\epsilon>0$, a rank-one projection $| x \rangle \langle x| \in D_n$, 
write $A_n = \Psi_n^c(|x \rangle \langle x|) - \tilde I_k$ with $\tilde I_k = I_k/k$, and follow the notation \eqref{eq:f_n}, to have
\eq{
\trace \left [A_n^2 \right] 
= \trace \left[\Phi_n^c(R|x \rangle \langle x|R) A_n \right] 
= \frac{1}{k}\langle x|R \left[ \sum_{i,j} (A_n)_{i,j} X_i^*X_j\right] R|x \rangle 
 \leq \frac{1+\epsilon}{k+1-2 \sqrt k} \cdot f_n(A_n)
}
for sufficiently large $n$. 
Here, the trace-preserving property of $\Psi^c$ implies that $\trace A_n = 0$.
Now, we will use uniform convergence: $f_n \to f$ in Lemma \ref{lemma:uniform},
and the Haagerup inequality in Theorem \ref{thm:Haagerup}.
Choose large enough $n$ and we have
\eq{\label{eq:app_Haa}
\trace \left [A_n^2 \right] 
\leq\frac{1+2\epsilon}{k+1-2 \sqrt k} \cdot f(A_n) 
\leq \frac{3(1+2\epsilon)}{k+1-2 \sqrt k} \,  \left\| A_n \right\|_2 \ ,
}
where $f(\cdot)$ is defined in \eqref{eq:f_def}.

Dividing the both sides by $\|A_n\|_2$ shows that 
\eq{
\max_{|x \rangle \langle x|   \in D_n}\|\Psi_n^c(|x \rangle \langle x|) - \tilde I_k \|_2 \leq \frac {3(1+2\epsilon)}{k+1-2\sqrt{k}} \ ,
}
which proves our first claim.

Secondly, then the above estimate implies that almost surely
\eq{ \label{eq:lower_bound}
\liminf_{n \to \infty} S_{\min} (\Psi_n^c)  \geq  \log k - \frac {9k}{(\sqrt{k}-1)^4}\approx \log k - \frac {9}{k} \ .
}
To this end, we use the following quadratic approximation of entropy; see \cite{has}. 
\eq{
S(\Psi_n^c(|x \rangle \langle x| )) \geq \log k - k \cdot \| \Psi_n^c(|x \rangle \langle x|) - \tilde I_k \|_2^2 \ .
}

Thirdly, by using the method in \cite{hw_vio}, we know that 
\eq{
\langle b_n| \left(\Psi_n^c \otimes \bar \Psi_n^c ( |  \langle b_n \rangle\langle b_n |) \right) | b_n \rangle
\geq \frac 1k
}
where $ | b_n \rangle$ is the Bell state in \eqref{eq:bell_state}.
This in turn implies that it always holds that
\eq{\label{eq:hw_trick}
S_{\min} (\Psi_n^c \otimes \bar \Psi_n^c) < 2 \log k - \frac {\log k}{k} + \frac 2 k \ .
}
The proof can be found in \cite[Lemma 2.1]{f-vio} with $a=1$, for example.

Finally, \eqref{eq:lower_bound} and \eqref{eq:hw_trick}, with the fact $S_{\min}(\Psi_n^c) = S_{\min}(\bar \Psi_n^c)$, prove the claim. 

\end{proof}
\begin{remark} 
The method with the trace-less condition, also used in \cite{collins2015estimates, collins_random_unitary},
requires trace-preserving property. This is why,
a similar method would not work for the CP maps in \eqref{eq:our-map}.
\end{remark}

\section{Discussions}
In this study, we treated some completely positive maps, which are almost trace-preserving asymptotically, but not exactly. 
Initially, we worked on random unitary channels or external fields in \eqref{eq:ruc}, and tried to get a similar formula as in Theorem \ref{theorem:bounding-condition}, which was unsuccessful.
To this end, one would have to study the support of the following quadratic forms of unitary elements:
\eq{\label{eq:quadratic_u}
\sum_{i,j} a_{i,j} u_i u_j^* \ .
}
This would be a generalization of the result on linear forms of unitary elements in \cite{akemann_ostrand}.
In \cite{belinschi2016almost} they made exact calculations on the counterparts for quantum channels which are generated with the isometries in \eqref{eq:iso} being truncation of the Haar-distributed unitary matrices.
Solving the above problem with \eqref{eq:quadratic_u} would lead us to exact evaluation of violation of MOE for random unitary channels or external fields. 
Solving the optimization problem in \eqref{eq:maximum_app} is also left open. 

In this paper, the dimensions of input and output spaces of the maps are the same, and the dimension of environment was fixed in the asymptotic processes. To alter these conditions, we would need completely new techniques of strong convergence.

\section*{Acknowledgement}
M.F. benefited from useful comments from Ion Nechita and Michael Wolf. 
Beno{\^{i}}t Collins is thanked for suggesting application of Haagerup's inequality during fruitful discussions.
Also, M.F. was financially supported by JSPS KAKENHI Grant Number JP16K00005, JP20K11667. T.H. is supported by JSPS Grant-in-Aid for Young Scientists (B) 15K17549 and 19K14546. M.F. and T.H. are both supported by JSPS and MAEDI Japan--France Integrated Action Program (SAKURA). 
This work was supported by 
Japan--France Integrated Action Program (SAKURA), Grant number JPJSBP120203202. This research is an outcome of Joint Seminar supported by JSPS and CNRS under the Japan-France Research Cooperative Program. The authors thank the anonymous referee for the insightful comments and suggestions. 

\appendix 

\section{Completely positive maps and their complement} 
\label{appendix:complementary}

A linear map $\Phi : \mathcal M_\ell (\mathbb C) \to   \mathcal M_n(\mathbb C) $ is called completely positive if $\Phi \otimes 1_{\mathbb C^m}$ is positive for $\forall m \in \mathbb N$, where $1_{\mathbb C^m}$ is the identity operator on $\mathcal M_m (\mathbb C)$.
In this case, $\Phi$ always has its corresponding Kraus representation:
\eq{
\Phi(\rho) = \sum_{i=1}^k X_i \rho X_i^*
}
where $A_i$'s are $n \times \ell$ matrices. 
Hence, one can define the following complementary map:
\eq{\label{eq:complement}
 \Phi^c  (\rho) = \sum_{i,j =1}^k  \trace[X_i \rho X_j^* ] |i \rangle\langle j|\ .
 }
Then, we have \cite{Holevo2005a}\cite{king2005properties}:
\begin{proposition}\label{proposition:complement}
For any vector $v \in \mathbb C^n$, 
$\Phi(|v \rangle \langle v|)$ and $\Phi^c(|v \rangle \langle v|)$ share
the non-zero eigenvalues.
In particular, $\| \Phi\|_{p} = \| \Phi^c \|_{p}$.
\end{proposition}
\begin{proof}
First set $k \times 1$ block matrix as
\eq{\label{eq:iso}
(V)_i = X_i
}
then we can write
\eq{
\Phi(|x \rangle \langle x|) = \trace_{\mathbb C^k} [V |x \rangle \langle x| V^* ]
\eqtext{and}
\Phi^c(|x \rangle \langle x|) = \trace_{\mathbb C^n} [V |x \rangle \langle x| V^* ] \ .
}
On the other hand, by the Schmidt decomposition we can write, assuming $k \leq n$ w.l.o.g.
\eq{
V|x\rangle = \sum_{i=1}^k \sqrt{\lambda_i} \, |e_i \rangle \otimes |f_i \rangle
}
where $|e_i \rangle$ and $|f_i \rangle$ are some orthonormal set of vectors in $\mathbb C^k$ and $\mathbb C^n$, respectively. 
Therefore, 
\eq{
\Phi(|x \rangle \langle x|) = \sum_{i=1}^k \lambda_i  |f_i \rangle \langle f_i |
\eqtext{and}
\Phi^c(|x \rangle \langle x|) =  \sum_{i=1}^k \lambda_i  |e_i \rangle \langle e_i | \ .
}
This completes the proof. 
\end{proof}
The complementary maps of $\Phi_n$'s in \eqref{eq:our-map} 
are rather easy to handle because the output space has the fixed dimension $k$
 while $n \to \infty$.

\section{Partially solving optimization problem}\label{section:lagrange}
We want to calculate 
\eq{\label{eq:maximum_app}
\max_{\substack{ A \in \hat D_k\\ \|A\|_q=1}} f(A), 
}
where $(A)_{i,j} = a_{i,j}$.
Recall from Theorem \ref{theorem:bounding-condition} that
\eq{\label{eq:minimum_app}
f(A) = 
\left \| \sum_{i,j =1}^k  a_{i,j} s_i s_j \right\| = \min_{x \in (0,1/\|A\|)} h(x,A), 
}
where  
\eq{
h(x,A) =\frac{1}{x} + \sum_{i=1}^k \frac{\lambda_i}{1-\lambda_i x} 
}
and $(\lambda_i)_{i=1}^k$ are the non-increasing sequence of the eigenvalues of $A$.

\begin{lemma}\label{lemma:two_kinds}
For $q\geq 2$ with $k \geq 2$, and for $1<q<2$ with large enough $k \in \N$, every matrix $A$ that gives the maximum in \eqref{eq:maximum_app} is of full rank and the eigenvalues are of at most two kinds.
\end{lemma}
\begin{proof}
First, we prove that all non-zero eigenvalues are of two kinds in the optimum.
Suppose the maximum of \eqref{eq:maximum_app} is achieved by some $A$:
\eq{
A = \mathrm{diag}(\lambda_1,\dots, \lambda_m,0,\ldots,0) \ ,
}
where $1>\lambda_1 \geq \cdots \geq \lambda_m >0$.
Here, we assume $3 \leq m \leq k$ because our claim is obvious otherwise.
Now, w.l.o.g, we assume that $A$ is positive definite $k \times k$ Hermitian matrix. Then, the maximum is achieved at an interior point of the $(m-1)$-dimensional compact surface $\{\lambda \in [0,1]^m: \|\lambda\|_q=1\}$, and hence we can apply the method of Lagrange multiplier. 

With Theorem \ref{theorem:bounding-condition}, one can apply the implicit function theorem to the critical point condition $h_x(x, A) = 0$ and get a unique function $x=x(\lambda_1, \ldots,\lambda_m)$ of $C^\infty$ class taking values in $(0,1/\|A\|)$.
Then, our Lagrange function becomes:
\eq{
L(\lambda_1,\ldots,\lambda_m)= h(x,A) - \mu \left( \sum_{i=1}^m\lambda_i^q-1 \right) \ ,
}
where $\mu$ is the Lagrange multiplier. 
If $(\lambda_1, ...,\lambda_m)$ is a point of the interior of the surface $ \sum_{i=1}^m\lambda_i^q = 1$ at which $h(x,A)$ achieves the maximum, then for some $\mu$ we must have  
\eqn{
L_{\lambda_j}(\lambda_1,\ldots,\lambda_m)
&=-\frac{x_{\lambda_j}}{x^2}+\sum_{i\neq j}\frac{\lambda_i^2x_{\lambda_j}}{(1-\lambda_ix)^2}+\frac{1+\lambda_j^2x_{\lambda_j}}{(1-\lambda_jx)^2}-\mu q\lambda_j^{q-1}\\
&=\lrb -\frac{1}{x^2}+\sum_{i=1}^m\frac{\lambda_i^2}{(1-\lambda_ix)^2}\rrb x_{\lambda_j}+\frac{1}{(1-\lambda_jx)^2}-\mu q\lambda_j^{q-1}\\
&=\frac{1}{(1-\lambda_jx)^2}-\mu q\lambda_j^{q-1}=0. 
}
This shows that $\mu>0$ and 
\eq{
1-\lambda_j x = (\mu q)^{-1/2} \lambda_j^{-(q-1)/2}. \label{eq:lambdas}
}
Set $a =  (\mu q)^{-1/2}$ and $b = (q-1)/2$. Then $\lambda_1,\dots, \lambda_m$ all satisfy the equation: 
\eq{\label{eq:two_kind}
\ell(\lambda):= a \lambda^{-b} + x \lambda -1 =0 \ .  
}
The derivative $\ell'(\lambda) = -a b \lambda^{-b-1}+x$ is negative for $\lambda \in (0, c)$ and positive for $x >c$, where $c= (ab /x)^{1/(1+b)}$. Thus the number of solutions $\lambda>0$ to \eqref{eq:two_kind} is at most two. 

Next, we prove that an optimal $A$ does not have zero eigenvalues. To this end, 
suppose the maximum of \eqref{eq:maximum_app} is achieved for some $A$:
\eq{\label{eq:A_two_non_zero}
A = \mathrm{diag}(\underbrace{\alpha,\dots,\alpha}_{r},\underbrace{\beta, \dots, \beta}_{m-r}, \underbrace{0,\ldots,0}_{k-m})
}
for $0\leq r < m \leq k$ and $0< \beta < \alpha$.
If $r=0$, then $\beta = m^{-1/q}$ and $h(x,A)$ achieves its minimum $m^{-1/q}(1+\sqrt{m})^2$ at $x=\frac{m^{1/q}}{1+\sqrt{m}}$. Extending $m$ to real numbers, let
\eq{
g(m):=  m^{-1/q} (1+\sqrt{m})^2
}
so that $g'(m) >0$ for all $m > 1$ whenever $q \geq 2$, which implies that $m=k$.
For $1<q<2$, one can choose large enough $k$ to conclude $m=k$, because $g(m) \approx m^{1-1/q}$ when $m$ is large.

Now, suppose $r \geq 1$.
The condition $\|A\|_q=1$ means $r \alpha^q + (m-r)\beta^q = 1$ and the condition $\alpha > \beta >0$ implies the constraint $1\leq r < \alpha^{-q} < m $. For fixed $\alpha \in (0,1)$ and $r \geq1$ with $r\alpha^q < 1 $ we regard $\beta$ as a function of $m \geq \alpha^{-q}$, where we extended $m$ to real numbers:
\eq{\label{eq:beta_m}
\beta= \beta(m) = \left(\frac{1-r\alpha^q}{m-r}\right)^{1/q}.   
}
We then rewrite $h(x,A)$ in terms of the parameter $m$ ($m$ now defines $A$):
\eq{
g(m,x) := \frac{1}{x} + \frac{r}{\frac{1}{\alpha}-x} + \frac{m-r}{\frac{1}{\beta}-x}
}
where $x \in (0, 1/\alpha)$ and $m \in[ \alpha^{-q},k]$.
Since
\eq{
\frac{\partial g}{\partial m}  = \frac{\frac{1-1/q}{\beta}-x}{(\frac{1}{\beta}-x)^2} ,
}
the condition
\eq{\label{eq:sufficient_full_full}
(1-1/q)\beta^{-1} \geq 1/\alpha
}
implies that $\partial g / \partial m >0$ for all $x \in (0,1/\alpha)$. 
Otherwise, the condition $(1-1/q)\beta^{-1} < 1/\alpha$ implies
\begin{equation}\label{eq:m}
\frac{\partial g}{\partial x}\left(m,\frac{1-1/q}{\beta}\right) >0 \ ,
\end{equation}
which would turn out to prove that $\partial g / \partial m >0$
in the neighborhood of $\arg \min_{x \in (0,1/\alpha)} g(m,x)$, because the minimum is achieved by $x \in (0, (1-1/q)\beta^{-1})$ in the view of Theorem \ref{theorem:bounding-condition}.
Therefore, since 
either \eqref{eq:sufficient_full_full} or \eqref{eq:m} is always true, $\min_{x \in (0,1/\alpha)} g(u,x)$ is a strictly increasing function of $m$. This means that the optimum should be realized for $m=k$ in \eqref{eq:A_two_non_zero}.
In particular, note that \eqref{eq:A_two_non_zero} satisfying \eqref{eq:sufficient_full_full} itself (and $m<k$) cannot give the optimum because $(1-1/q)\beta^{-1}$ increases in $m$.

Now we prove \eqref{eq:m} assuming that \eqref{eq:sufficient_full_full} does not hold.
Direct calculations give
\eq{
\frac{\partial g}{\partial x}=-\frac{1}{x^2} + \frac{r}{(\frac{1}{\alpha}-x)^2} + \frac{m-r}{(\frac{1}{\beta}-x)^2} \ ,
}
and by using the condition $(1-1/q)\beta^{-1} < 1/\alpha$ we have
\eq{
\frac{\partial g}{\partial x}\left(m,\frac{1-1/q}{\beta}\right) 
&= -\frac{\beta^2}{(1-1/q)^2} + \frac{r}{(\frac{1}{\alpha}-\frac{1-1/q}{\beta})^2} + \frac{m-r}{(\frac{1}{\beta}-\frac{1-1/q}{\beta})^2}\\
&\geq -\frac{\beta^2}{(1-1/q)^2} + \frac{r}{(\frac{1}{\beta}-\frac{1-1/q}{\beta})^2} + \frac{m-r}{(\frac{1}{\beta}-\frac{1-1/q}{\beta})^2}\\
& = (q\beta)^2 \left( m-\frac{1}{(q-1)^2} \right), 
}
which is positive for all real $m \geq \alpha^{-q}$ if $q \geq2$. 
In particular, note that \eqref{eq:A_two_non_zero} satisfying the condition 
\eq{\label{eq:sufficient_full_partial}
m>(q-1)^{-2}
}
itself cannot give the optimum as one always can increase $m$ to satisfies the same condition. 

Let us consider the case $1<q<2$.
To this end we assume two conditions. One is that $(1-r\alpha^q)^{1/q} \geq \alpha(1-1/q)$, because otherwise we would get via \eqref{eq:beta_m} 
\eq{
(1-r\alpha^q)^{1/q} \leq \alpha(1-1/q)(m-r)^{1/q} 
\quad \text{which means} \quad
\beta \leq \alpha (1-1/q) \ ,
}
satisfying the condition \eqref{eq:sufficient_full_full}.
The second is that $r \leq (q-1)^{-2}$, because otherwise, as $r<m$, \eqref{eq:sufficient_full_partial} would be satisfied and there would be nothing to prove.
Then, we calculate
\eq{
1-r\alpha^q \geq \frac{(1-1/q)^q}{r+(1-1/q)^q} \geq \frac{(1-1/q)^q}{1/(q-1)^2+(1-1/q)^q} =: C_q
}
where we used the above two conditions for the two inequalities, respectively.
This gives a lower-bound for $g(m,x)$:
\eq{
g(m,x) \geq \frac{m-r}{\frac{1}{\beta}} = (1-r\alpha^q)^{1/q}(m-r)^{1-1/q}\geq  \ce_q^{1/q} \left(m-(q-1)^{-2}\right)^{1-1/q} =: \tilde C_{q,m}
\ ,
}
which in particular holds for $m=k$.
Note that this bound depends only on $q$ and $m$. 
On the other hand, we get an upper-bound for all $m\leq (q-1)^{-2}$, which will be shown below; again
we only consider this case otherwise \eqref{eq:sufficient_full_partial} would be satisfied.
Since for $0<x<1/\alpha$
\eq{
g(m,x) \leq \frac{1}{x} + \frac{m}{\frac{1}{\alpha}-x} \ ,
}
we have
\eq{
\min_{0<x<1/\alpha} g(m,x) \leq \alpha(1+\sqrt{m})^2 \ .
}
This in turn implies that  
\eq{
\max_{2 \leq m  \leq  1/(q-1)^2}\min_{0<x<1/\alpha}g(m,x) \leq \frac{\alpha q^2}{(q-1)^2} \ .
}
Therefore, for large enough $k$ we have for all $m\leq (q-1)^{-2}$
\eq{\min_{0<x<1/\alpha}g(m,x) \leq
\frac{\alpha q^2}{(q-1)^2}<
\tilde C_{q,k} \leq \min_{0<x<1/\alpha}g(k,x) 
}
This completes the proof.
\end{proof}

Now, we state the main lemma in this section:
\begin{lemma}\label{lemma:optimal-shape}
Every matrix $A$ that gives the maximum in \eqref{eq:maximum_app} has eigenvalues of the form $(\alpha,\beta,\beta,\dots,\beta)$, where $\alpha \geq \beta >0$ for
$q\geq3$ with $k\geq2$, or for $1<q<3$ with large enough $k$. 
\end{lemma}
 
\begin{proof}
 By Lemma \ref{lemma:two_kinds}, we know that the maximum of \eqref{eq:maximum_app} is achieved by some $A$ in the form of
\eq{
A = \mathrm{diag}(\underbrace{\alpha,\dots,\alpha}_{r},\underbrace{\beta, \dots, \beta}_{k-r}) \ ,
}
where $1\leq r < k$ and $0< \beta < \alpha$, unless $A$ is proportional to the identity. Similarly as before, fix $\alpha \in(0,1)$ and $k >\alpha^{-q}$, and think of $\beta$ as a function of real numbers $r \in [1,k)$:
\eq{
\beta = \beta(r) = \left( \frac{1-r\alpha^q}{k-r} \right)^{1/q} \ .
}
via $r \alpha^q + (k-r)\beta^q = 1$, and rewrite $h(x,A)$ in terms of the parameter $r$:
\eq{
g(r,x) = \frac{1}{x} + \frac{r}{\frac{1}{\alpha}-x} + \frac{k-r}{\frac{1}{\beta}-x} \ .
}
Some calculations show that 
\eq{
\frac{\partial g}{\partial r}=  \frac{a_{q+1} x - a_q}{q \beta^{q-1}(1-\alpha x)(1-\beta x)^2}  \ ,
}
where 
\eq{
a_q= \alpha^q - \beta^q - q\beta^{q-1} (\alpha -\beta)
\equiv  (q-1) \beta^q - q \alpha \beta^{q-1} + \alpha^p \ .
}
Moreover, one can see that $a_q >0$, and hence, $\partial g / \partial r<0$ if and only if $x< a_q/a_{q+1}$. A few lines of calculations show that 
\eq{\label{eq:alpha_side}
\frac{1}{\alpha}-\frac{a_q}{a_{q+1}} = \frac{q \beta^{q-1}(\alpha-\beta)^2}{\alpha a_{q+1}} >0, 
}
and hence $a_q/a_{q+1}$ lies in the interval $(0,1/\alpha)$. As before, if we  prove that 
\eq{\label{eq:one_large_single_condition}
\frac{\partial g}{\partial x}\left(r, \frac{a_q}{a_{q+1}}\right) >0
}
for $r \in [1,k)$
then $\min_{x\in (0,1/\alpha)} g(r,x)$ is a strictly decreasing function of $r \in [1,k)$, and hence our claim will be proved. 
Since 
\begin{align}
\frac{\partial g}{\partial x}\left(r, \frac{a_q}{a_{q+1}}\right)
&\geq -\frac{a_{q+1}^2}{a_q^2} + \frac{r}{(\frac{1}{\alpha}-\frac{a_q}{a_{q+1}})^2}= a_{q+1}^2\left(-\frac{1}{a_q^2} + \left(\frac{\sqrt{r}\alpha}{q\beta^{q-1}(\alpha-\beta)^2}\right)^2 \right), \label{eq:no_zero_ev}
\end{align}
it suffices to prove that $a_q > q \beta^{q-1}(\alpha-\beta)^2/\alpha$. With the new variable $u=\beta / \alpha$, this inequality is equivalent to 
\eq{
1- 2q u^{q-1} +(3q-1)u^q -q u^{q+1} >0, 
}
which is valid by calculus for all $u\in(0,1)$ whenever $q\geq 3$. 

Now, there remains to prove \eqref{eq:one_large_single_condition} for $1<q<3$.
To this end, first notice that 
\eq{\label{eq:beta_side}
\frac{1}{\beta} -\frac{a_q}{a_{q+1}} = \frac{(\alpha-\beta)(\alpha^q-\beta^q)}{\beta a_{q+1}}
\leq \frac{q \alpha^{q-1}(\alpha-\beta)^2}{\beta a_{q+1}}
}
where mean value theorem was used to the inequality. 
Next, using \eqref{eq:alpha_side} and \eqref{eq:beta_side} we get
\eq{
\frac{\partial g}{\partial x}\left(r, \frac{a_q}{a_{q+1}}\right)
&= -\frac{a_{q+1}^2}{a_q^2} + \frac{r}{(\frac{1}{\alpha}-\frac{a_q}{a_{q+1}})^2} + \frac{k-r}{(\frac{1}{\beta}-\frac{a_q}{a_{q+1}})^2} \\
&\geq a_{q+1}^2 \left[ 
-\left(\frac{1}{a_q}\right)^2 
+ r \left(\frac{\alpha}{q\beta^{q-1}(\alpha-\beta)^2} \right)^2
+ (k-r) \left(\frac{\beta}{q\alpha^{q-1}(\alpha-\beta)^2} \right)^2
\right]\\
& = a_{q+1}^2 \alpha^{-2q} \left[ 
-\left(\frac{1}{(q-1)u^q - q u^{q-1} +1} \right)^2
+ \frac{ku^{q+1} + r(1-u^{q+1})}{(q(1-u)^2u^{q-1})^2}
\right]\\
& \geq a_{q+1}^2 \alpha^{-2q} \left[ 
-\left(\frac{1}{(q-1)u^q - q u^{q-1} +1} \right)^2
+ \frac{ku^{q+1} + 1-u^{q+1}}{(q(1-u)^2u^{q-1})^2}
\right]
}
where $u = \beta/\alpha$.
Hence, we now prove
\eq{\label{eq:one_large_proof_last}
(q-1)u^q - q u^{q-1}+1 > \frac{q(1-u)^2 u^{q-1}}{\sqrt{1+(k-1)u^{2q}}} \ .
}
To this end, fix $1<q<3$ and $0< \delta < 1$ for now.
While it is obvious when $u$ is close to $0$,
we have to think of the following asymptotic behaviours:
\eq{
(LHS) = \frac{q(q-1)}{2} (1-u)^2 (1 + o(1))
\qquad 
(RHS) = \frac{q}{\sqrt{k}} (1-u)^2 (1 + o(1)) \ .
}
when $u$ is close to $1$.
With this, \eqref{eq:one_large_proof_last} holds on
$(0, \delta] \cup [1- \delta, 1)$ for large enough $k$.
On the other hand, 
\eq{
\min_{u \in [\delta, 1-\delta]} (q-1)u^q - q u^{q-1}+1
> \frac{q}{\sqrt{1+(k-1)\delta^{2q}}} 
>\max_{u \in [\delta, 1-\delta]}\frac{q(1-u)^2 u^{q-1}}{\sqrt{1+(k-1)u^{2q}}} 
}
holds for large enough $k$, where the second is always true. 
As a whole, for fixed $1<q<3$ and $0< \delta < 1$, \eqref{eq:one_large_proof_last} holds for large enough $k$. This completes the proof.
\end{proof}

\section{H\"older's inequality for positive semi-definite matrices}
In this appendix, we discuss H\"older's inequality for positive semi-definite matrices, which plays a key role in Theorem \ref{theorem:p-norm-convergence} and Section \ref{section:lagrange}:
\begin{lemma}\label{lemma:duality}
For a $k \times k$ Hermitian matrix $A \geq 0$ with $A \not = 0$ and $p,q>1$ with $\frac 1p + \frac 1q =1$, we have
\eq{\label{eq:dual}
\|A\|_p = \max \left\{ \trace [AB] :B \geq 0, \, \|B\|_q =1 \right\} \ .
}
Moreover, the maximum is achieved if and only if for all $i \in \{ 1, \ldots, k \}$,
\eq{\label{eq:pq-relation}
\lambda_i(B)^q = \|A\|_p^{-p} \lambda_i(A)^p \ ,
}
where $A$ and $B$ share the same eigenspaces with respect to the eigenvalues $\lambda_i(\cdot)$. 
\end{lemma}
\begin{proof}
In this proof, we follow the notations $d(\cdot)$ and $\lambda(\cdot)$ defined in Lemma \ref{lemma:birkhoff}.

Choose the coordinate so that $B$ is diagonal and the diagonal elements, which are the eigenvalues, appear in the non-increasing order. Then, using H\"older's  inequality and Lemma \ref{lemma:birkhoff} we have
\eq{\label{eq:hoelder}
\trace [AB] = \sum_{i=1}^k a_{i,i} \, \lambda_i(B)
\leq \|d(A)\|_p\|\lambda(B)\|_q \leq \|A\|_p\|B\|_q
}

Moreover, by using the equality conditions for H\"older's inequality and Lemma \ref{lemma:birkhoff}, the equality in \eqref{eq:hoelder} holds if and only if for some constant $c>0$,
\eq{
a_{i,i} = c \lambda_i(B)^{q-1}, \text{ and $A$ is diagonal with }  d(A)=\lambda(A) \ .
}
This completes the proof.
\end{proof}

\begin{lemma}\label{lemma:birkhoff}
Let $1<p<\infty$. For a positive semi-definite $k \times k$ Hermitian matrix $A$, it holds that
\eq{\label{eq:shatten}
\|d(A)\|_p \leq \|\lambda(A)\|_p = \|A\|_p
}
and the equality holds if and only if $A$ is diagonal and $d(A)=\lambda(A)$.
Here,
\eq{
d(A) = (d_1(A), \ldots, d_k(A))^T \qquad \text{and} \qquad
\lambda(A) = (\lambda_1(A), \ldots, \lambda_k(A))^T
}
are the diagonal elements and the eigenvalues of $A$, respectively, 
which are both in the non-increasing order. 
Note that we used the notation of Schatten $p$-norm $\| \cdot \|_p$ for vectors and matrices. 
\end{lemma}
\begin{proof}
Through diagonalization of $A$, we define a doubly stochastic matrix $T$ relating $\lambda(A)$ and $d(A)$:
\eq{
A = U \Lambda U^* \eqtext{and} 
d(A) = T \lambda(A) \eqtext{where} (T)_{i,j} = |(U)_{i,j}|^2\ .
}
Here, $\Lambda$ is the diagonal matrix such that $\Lambda = \mathrm{ diag}(\lambda (A))$. 
Via Birkhoff's theorem \cite[Theorem II.2.3]{bhatia2013matrix},
which states that the set of doubly stochastic matrices is the convex hull of the permutation matrices (they are extreme points),
\eq{\label{eq:convex_hull}
d(A) = \sum_{j=1}^N t_j P_j \lambda(A)
}
for some coordinate-permutations $\{P_j\}$ and a probability distribution $\{t_j\}$. This means that 
\eq{
\|d(A)\|_p \leq \sum_{j=1}^N t_j  \|P_j \lambda(A)\|_p = \|A\|_p \ . 
}
This yields the inequality in \eqref{eq:shatten}.
Moreover, the equality condition implies $d(A) = P_j \lambda(A)$ for a certain $j$, i.e. $d(A) = \lambda(A)$. Also in this case $A$ is diagonal.
Otherwise, pick a non-zero off-diagonal element $a_{m,n}$ with $m<n$ where $m$ is the smallest among such entries. Then we can decompose $A$ as 
\[
A=\begin{pmatrix} A_{m-1} & O  \\ O & A_{k-m+1}' \end{pmatrix}, 
\]
where $A_{m-1}$ and $A_{k-m+1}'$ are respectively $(m-1)\times (m-1)$ and $(k-m+1)\times(k-m+1)$ sub-matrices of $A$. Since $A_{m-1}$ is diagonal, the eigenvalues of $A_{k-m+1}'$ coincide with $\lambda_{m}\ge \lambda_{m+1} \ge \cdots \ge \lambda_k$. Now take the following $2 \times 2$ sub-matrix of $A_{k-m+1}'$ and calculate its characteristic equation:
\eq{
A':=\begin{pmatrix}
\lambda_m(A) & a_{m,n} \\a_{n,m} & \lambda_n(A)
\end{pmatrix},
\qquad
(x - \lambda_m(A))(x - \lambda_n(A)) = |a_{m,n}|^2 
}
so that the largest eigenvalue of $A'$ is larger than $\lambda_m(A)$.
However, the eigenvalues of $A'$ cannot be larger than $\lambda_m(A)$ according to Cauchy's interlacing theorem \cite[Corollary III.1.5]{bhatia2013matrix}, yielding a contradiction. This completes the proof.
\end{proof}

\section{About the $p$-norm in the case $1<p\leq 1.5$}
\label{sec:<1.5}
In Sections \ref{sec:vilation_phi} and \ref{sec:violation} we proved the violation of multiplicativity of MO$p$N for both $\Phi_n$ and its rectification $\Psi_n$ as long as $p>1.5$, and the violation of additivity of MOE (corresponding to $p=1$) for $\Psi_n$. In this appendix we briefly discuss the missing case $1<p\leq 1.5$. 

From a technical point of view, 
 the violation of multiplicativity of $p$-norm is not clear 
 for $1<p\leq 1.5$, because of the correction term $k^{-1/2}$, which often appeared in our analysis, 
e.g. in the square brackets in the first line of \eqref{eq:upper_bound_pair_multi}. This term is negligible for $p>1.5$ but not for $1\le p\leq 1.5$.
As long as the quantum channels $\Psi_n$ are concerned, we think that this is only a technical issue and conjecture that violation of multiplicativity is the case even for $1<p\leq 1.5$ (see \cite{hw_vio, aubrun2010nonadditivity} for the case of GEs with $1<p\leq 1.5$). 

Another interesting aspect for $1<p\leq 1.5$ is the optimality of Bell states for the conjugate pairs $ \Phi_n\otimes \bar \Phi_n$. Let $q$ be the number such that $1/p + 1/q=1$. By using the definition of $f(\cdot)$ in \eqref{eq:f_def} and the equality in \eqref{eq:max-f},
\eq{
\lim_{n\to\infty} \| \Phi_n \|_p 
= \lim_{n\to\infty} \| \Phi_n^c \|_p 
\geq \frac{1}{k} f(k^{-1/q} I) 
= k^{-1-1/q} \left \| \sum_{i=1}^k  s_i^2 \right\|
\geq  k^{-2+1/p}(\sqrt{k}+1)^2 \ . \label{D1}
}
Now, we compare this bound with the RHS of \eqref{eq:bell-input} as follows.
\eq{\label{eq:additivity_check}
&\left( k^{-2+1/p}(\sqrt{k}+1)^2 \right)^{2p} - \left[ 
\left( \frac 1k + \frac 1{k^2} \right)^p + (k^2-1) \left( \frac 1 {k^2}\right)^p \right] \\
& = k^{-2p} \left( k^2(1+k^{-1/2})^{4p} -(k+1)^p -k^2+1\right) \ .
}
Next, the inequality $(1+x)^a > 1+a x$ for $a\geq1$ and $x> 0$ implies 
\eq{\label{eq:additivity_check2}
k^2(1+k^{-1/2})^{4p} -(k+1)^p -k^2+1 
&> k^2(1+4p k^{-1/2}) -(1+k)^p -k^2+1\\
&= 4 p k^{3/2} - (1+k)^p +1 >0 \ .
}
To see the last inequality, 
let
$g(k):= 4 p k^{3/2} - (1+k)^p +1$ so that for $1<p\leq 1.5$
\eq{
g(1)&=4p -2^p+1 \geq 4 - 2^{3/2}+1>0 \quad \text{and}\\
g'(k) &= 6p \sqrt{k} -p(1+k)^{p-1} \geq p[6\sqrt{k} -\sqrt{1+k}] > 0 \ .
}
Hence, for Bell states $|b_n \rangle$ almost surely we have
\eq{
\lim_{n \to \infty} \left\| \Phi_n\otimes \bar \Phi_n (|b_n\rangle \langle b_n|) \right\|_p < \lim_{n \to \infty} \left\| \Phi_n \right\|_p \left\|  \bar \Phi_n \right\|_p \ .
}
On the other hand, it always holds that
$
\left\| \Phi_n\otimes \bar \Phi_n \right\|_p \geq \left\| \Phi_n \right\|_p \left\|  \bar \Phi_n \right\|_p, 
$
so that 
\eq{
\lim_{n \to \infty} \left\| \Phi_n\otimes \bar \Phi_n (|b_n\rangle \langle b_n|) \right\|_p < 
\liminf_{n \to \infty} \left\| \Phi_n\otimes \bar \Phi_n \right\|_p \ . \label{eq:Bell_not_enough}
}
The inequality \eqref{eq:Bell_not_enough} shows that the Bell states are not optimal inputs for the $p$-norm of the conjugate pairs $ \Phi_n\otimes \bar \Phi_n$ when $1<p\leq 1.5$. This may seem against the evidence stated in \cite{CFN_bell, CFN_w_col} that Bell states are optimal inputs among reasonably entangled inputs for conjugate pairs of random quantum channels. However note that
the maps $\Phi_n$ are not trace-preserve, and hence not quantum channels. 
We believe that the Bell states are optimal inputs for the $p$-norm of the conjugate pairs of quantum channels $\Psi_n\otimes \bar \Psi_n$, though. 

To conclude this appendix, we point out that the above arguments would fail at the steps \eqref{eq:additivity_check} and \eqref{eq:additivity_check2} for the random quantum channels $\mathcal E$ in \eqref{eq:ruc} with equal weights. Indeed,  \eqref{eq:additivity_check2} would be
\eq{
\left( k^{-1+1/p} \right)^{2p} - \left[ k^{-p} + k^{-2p}(k^2- k) \right] 
= k^{-2p} \left[ k^2 - k^{p} - k^2 + k  \right]
< 0 \ ,
}
where we mimicked the computations in the proof of Theorem \ref{theorem:bell-input} using the fact that
\eq{
\varphi(u_i u_j^* u_s u_t^*) = 
\begin{cases}
1& (i=j \text{ and } s=t) \vee(i=t \text{ and } j=s)  \\
0& \text{otherwise}
\end{cases}
}
for free Haar unitary elements 
$\{u_i\}_{i=1}^k$. 
Indeed, similarly as in \eqref{eq:bell_limit},
\eq{
 \left(\Omega_n^c \otimes \bar \Omega_n^c \right)  (|b_n \rangle \langle b_n|)
& \to \frac 1{k^2} \left[ 
\sum_{\substack{i,s=1\\i \not =s}}^k 
|i,i\rangle \langle s,s| 
+\sum_{\substack{i,j=1\\i \not =j}}^k 
|i,j\rangle \langle i,j| 
+\sum_{\substack{i=1}}^k
|i,i\rangle \langle i,i| 
\right]  \qquad (n \to \infty) \\
&= \frac 1k |b_k \rangle \langle b_k| 
+ \frac{1}{k^2}\sum_{\substack{i,j=1\\i \not =j}}^k 
|i,j\rangle \langle i,j| 
\ ,
}
where the first and second Hermitian matrices are orthogonal to each other.

\bibliographystyle{plain}

\bibliography{ref_sce_vio}

\end{document}